\documentclass[letter,11pt]{article}
\usepackage[utf8]{inputenc}
\usepackage{amsfonts}
\usepackage{amssymb}
\usepackage{amsthm}
\usepackage{amsmath}
\usepackage[letterpaper,left=1in,right=1in,top=1in,bottom=1.2in,footskip=0.6in]{geometry}
\usepackage[all]{xy}
\usepackage{marginnote}
\usepackage{comment}

\usepackage{hyperref}
\hypersetup{colorlinks=true,citecolor=blue,linkcolor=red,urlcolor=blue}

\usepackage{eurosym}
\usepackage{graphicx}
\usepackage{colortbl}
\usepackage{epstopdf}
\usepackage{array}
\usepackage[usenames,dvipsnames]{xcolor}

\usepackage{arydshln}

\newcommand*{\N}{\mathbb{N}}
\newcommand*{\C}{\mathbb{C}}

\newtheorem{theorem}{Theorem}
\newtheorem{lemma}[theorem]{Lemma}
\newtheorem{definition}[theorem]{Definition}
\newtheorem{corollary}[theorem]{Corollary}
\theoremstyle{definition}
\newtheorem{remark}{Remark}

\newtheorem{conjecture}{Conjecture}
\title{\bf Strong Spatial Mixing for General 2-Spin Systems:\\
A Unified Approach from Zero-Freeness}
\date{}

\author{Shuai Shao\thanks{School of Computer Science and Technology \& Hefei National Laboratory, University of Science and Technology of China. Supported by the Innovation Program for Quantum Science and Technology, 2021ZD0302901.}\\ 
{\tt  shao10@ustc.edu.cn}
\and Xiaowei Ye\thanks{\'Ecole Polytechnique. This work was done while Xiaowei Ye was an undergraduate student at the School of the Gifted Young, University of Science and Technology of China.}\\
\tt {yxwustc@mail.ustc.edu.cn}}

\begin{document}
\begin{titlepage}
    \maketitle    \thispagestyle{empty}
\begin{abstract}
We study the algorithmic implications of zero-free regions for the partition functions of 2-spin systems. While Barvinok's algorithm yields FPTASes in such regions, the applicability of Weitz's algorithm is limited to parameter regimes where strong spatial mixing (SSM) can be established. It remains open whether Weitz's algorithm can be applied to general zero-free regions, particularly in settings where no standard tree-recurrence-based proof of SSM is known.

In this paper, we establish new SSM results and thereby extend the applicability of Weitz's FPTAS  to all currently known zero-free regions of 2-spin systems with pinned vertices. 
We achieve this through a unified approach to deriving SSM from zero-freeness in the most general settings of 2-spin systems.
Our work features two key innovations. 
\begin{enumerate}
    \item {\bf New SSM results and Weitz's FPTASes beyond tree recurrences.} Our SSM results cover parts of the celebrated Lee-Yang zero-free region for the ferromagnetic Ising model, where \emph{no} tree-recurrence-based proof of SSM is currently known, or considered feasible. 
    The tree recurrence method typically relies on carefully designed potential functions, the construction and analysis of which can be highly challenging. 
For ferromagnetic 2-spin systems, it remains an open challenge whether such potential functions can be constructed.
We circumvent this difficulty by deriving SSM directly from zero-freeness. 
\item {\bf A novel and unified approach beyond cluster expansions.} The prior approach to deriving SSM from zero-freeness relies on  an analytic tool, namely cluster expansions, 
which are 
 model-specific and known only for a few restricted parameter settings such as the hard-core model near vertex activity $\lambda=1$. 
 A crucial limitation of this approach is that beyond these model-specific settings,  general cluster expansions are \emph{unknown} for 2-spin systems on general graphs.
We overcome this obstacle by introducing a purely combinatorial approach based on a novel {Christoffel–Darboux-type} identity that holds universally for 2-spin systems.
This provides a broadly applicable framework\footnote{In a recent follow-up work~\cite{shao2025zero}, this framework has been further extended to more general settings, including hypergraph independence polynomials, random cluster models, Potts models, and Holant problems, thereby yielding new SSM results in these settings.} for handling general 2-spin systems with arbitrary multivariate parameters and zero-free regions of arbitrary shape in a unified manner.
\end{enumerate}
\end{abstract}

\end{titlepage}
\newpage

\section{Introduction}
{Spin systems} arise from statistical physics to model interactions between
neighbors on graphs.
In this paper, we consider 2-spin systems.
Such a system is defined on a simple finite undirected graph $G=(V,E)$  and in this way the individual entities comprising the system correspond to the
vertices $V$ and their pairwise interactions correspond to the edges $E$.
The system is  associated with three parameters consisting of two edge activities $\beta$ and $\gamma$ that model the tendency of vertices to agree and disagree with their neighbors, and a vertex activity $\lambda$ that models a uniform external field which determines the propensity of a vertex to be the spin $+$.
A partial configuration of this system refers to a mapping $\sigma: \Lambda\to\{+,-\}$ for some $\Lambda\subseteq V$ which may be empty. It assigns one of the two spins $+$ and $-$ to each vertex in $\Lambda$.
When $\Lambda=V$, it is a configuration, and its weight denoted by $w(\sigma)$ is $\beta^{m_+(\sigma)}\gamma^{m_-(\sigma)}\lambda^{n_+(\sigma)}$, where $m_+(\sigma),m_-(\sigma)$ and $n_+(\sigma)$ denote respectively the number of $(+,+)$ edges, $(-,-)$ edges and vertices with the spin $+$. 
The \emph{partition function} of a 2-spin system is defined to be
\[Z_G(\beta,\gamma,\lambda):=\sum\limits_{\sigma:V\to\{+,-\}}w(\sigma).\] 
We also define the partition function   conditioning on a pre-described partial configuration $\sigma_\Lambda$ (i.e., each vertex in $\Lambda$, called a pinned vertex, is pinned to be the spin $+$ or $-$) denoted by $Z^{\sigma_\Lambda}_G(\beta,\gamma,\lambda)$ to be $\sum_{\sigma:V\to\{+,-\}\atop     \sigma|_{\Lambda}=\sigma_\Lambda}w(\sigma)$ where $\sigma|_{\Lambda}$ denotes the restriction of the configuration $\sigma$ on $\Lambda$.

Computing the partition function of the 2-spin system given an input graph $G$ is a very basic counting problem, and it is known to be \#P-hard for all complex-valued parameters $(\beta, \gamma, \lambda)$
except for a few very restricted settings such as $\beta\gamma=1$ or $\lambda=0$~\cite{bar82,Caiphard,clx-boolean-csp-complex}.
Many natural combinatorial problems can be formulated as computing the partition functions of 2-spin systems. 
For example, when $\beta = 0$ and $\gamma=1$,  $Z_G(0, 1, \lambda)$ is the independence polynomial of the graph $G$
(also known as the \emph{hard-core model} in statistical physics); it counts the number of independent 
sets of the graph $G$ when $\lambda =1$.
When $\beta=\gamma$, such a 2-spin system is the famous \emph{Ising model}.

In classical statistical mechanics, the parameters $(\beta, \gamma, \lambda)$ are usually non-negative real numbers and $(\beta, \gamma)\neq (0, 0)$.
Such 2-spin systems are divided into the \emph{ferromagnetic} case ($\beta\gamma>1$) and the \emph{antiferromagnetic} case ($\beta\gamma <1$).
For non-negative $(\beta, \gamma, \lambda)$ that are not all zero,
the partition function can be viewed as the normalizing factor of the Gibbs distribution, 
which is the distribution where a configuration $\sigma$ is drawn with probability 
${\rm Pr}_{G; \beta, \gamma, \lambda}(\sigma)=\frac{w(\sigma)}{Z_G(\beta, \gamma, \lambda)}$.
 However, it is meaningful to consider parameters with complex values. 
%First, the parameters are generally complex-valued for quantum computation. 
First, complex-valued parameters arise when we consider quantum computation. 
For instance, the partition functions of 2-spin systems with complex parameters are closely related to the output probability amplitudes of quantum circuits~\cite{de2011quantum,iblisdir2014low,mann2019approximation}. 
 Moreover, even in classical theory, the study of the location of \emph{complex} zeros of the partition function $Z_G(\beta, \gamma, \lambda)$ connects closely to the analyticity of the free energy $\log Z_G(\beta, \gamma, \lambda)$, which is a classical notion in statistical physics for defining and understanding  the phenomenon of \emph{phase transitions}.
 One of the first and also the best-known results regarding the zeros of  partition functions is the Lee-Yang theorem \cite{LeeYang1,LeeYang2} for the {Ising model}.
 This result was later extended to more general models \cite{asano1970lee,ruelle1971extension,simon1973varphi,newman1974zeros,lieb1981general}.
 Another standard notion for formalizing (the absence of) phase transitions in 2-spin systems is \emph{correlation decay}, which 
 %refers to that 
 means that 
 correlations between spins decay exponentially as a function of the distance between them. 
  More formally, given an infinite graph, such as an infinite regular tree, 
    one can characterize regions in the model’s parameter space where, under the associated Gibbs distribution, correlations between spins decay exponentially with distance. 

The two notions for (the absence of) phase transitions can also be exploited directly to devise fully polynomial-time deterministic approximation schemes (FPTASes) for computing the partition functions of 2-spin systems. 
The  method associated with {correlation decay}, or more precisely, \emph{strong spatial mixing} (SSM)\footnote{
There are weak spatial mixing and strong spatial mixing.
Weak spatial mixing just refers to the  correlation decay property; 
strong spatial mixing requires that correlations between spins decay exponentially with distance even when conditioning on a pre-described partial configuration.},
was originally developed by Weitz \cite{Weitz06},  and Bandyopadhyay and Gamarnik~\cite{BG-add-here}  for the hard-core model, and was extended to complex parameters~\cite{harvey2018computing}. 
It turns out to be a very powerful tool for antiferromagnetic 2-spin systems \cite{zhang2009approximating,li2012approximate,li2013correlation,sinclair2014approximation}.
In contrast, for ferromagnetic 2-spin systems, the correlation decay method has yielded only limited results~\cite{zhang2009approximating,guo2018uniqueness}.
In these cases, correlation decay is typically established via tree recurrence arguments, which rely on the uniqueness condition on regular trees (also known as the contraction condition~\cite{shao2019contraction}) and on the careful design of a suitable potential function.

The method turning complex zero-free regions of  partition functions into FPTASes was developed by Barvinok \cite{BarvinokBook} and extended by Patel and Regts \cite{PatelRegts17}. It is usually called the \emph{Taylor interpolation} method.
Motivated by this method, several  complex zero-free regions have been obtained for hard-core models \cite{PetersRegts19,Bencszero}, Ising models \cite{mann2019approximation,LSSFisherzeros,petersregts20,GGHP22,PRS23}, and general 2-spin systems \cite{zerofe,shao2019contraction}. 
Table~\ref{tab:my_label} summarizes the known zero-free regions for 2-spin systems.
For a subset $I\subseteq \mathbb{C}$, we use $\mathcal{N}(I)$ to denote a complex neighborhood of $I$. 
In the table, the fourth column (named ``Pinned'') indicates  whether the corresponding zero-free  regions remain valid for graphs with pinned vertices, while 
the fifth column (named ``Bounded'') indicates whether the result holds only for graphs with maximum degree bounded by some integer $d$.
The result in Row 8 corresponds to the celebrated Lee–Yang theorem~\cite{LeeYang2}. Furthermore, the regions in Rows 2, 4, and 6 are all contained within the region described in Row 7.

\begin{table}[!htpb]\small
\renewcommand\arraystretch{1.5}
    \centering
\begin{tabular}{|c|c|c|c|c|c|}
\hline
  & Model & Fixed parameters  & Pinned  & Bounded & Zero-free regions\\
\hline
 1& Hard-core & $\beta=0,\gamma=1$ & \bf Yes & Yes & $\{\lambda\in \mathbb{C}\mid |\lambda|<\frac{(d-1)^{d-1}}{{d}^{d}}\}$ \cite{PetersRegts19}\\
 2&  Hard-core  & $\beta=0,\gamma=1$ & \bf Yes &  Yes &  $\lambda\in \mathcal{N}(I)$, $I=[0, \frac{(d-1)^{d-1}}{{(d-2)}^{d}})$ \cite{PetersRegts19}\\
 3&  Hard-core & $\beta=0,\gamma=1$ & \bf Yes & Yes & $\lambda\in  \mathcal{N}(0)$ irregular shape \cite{Bencszero}\\
 4& Ising & $\lambda=1$ & \bf Yes & Yes &  $\beta \in  \mathcal{N}(I)$, $I=(\frac{d-2}{d},\frac{d}{d-2})$  \cite{LSSFisherzeros}\\
 5& Ising & $\beta\in(\frac{d-2}{d},1)$ & \bf Yes & Yes & $\{\lambda\in \mathbb{C}\mid |\arg(\lambda)|<\theta(\beta)\}$  \cite{petersregts20}\footnotemark\\
  6& Ising & $\beta\in(\frac{d-2}{d},1)$ & \bf Yes & Yes &  $\lambda \in  \mathcal{N}(I)$, 
  %$I=\{x>0\mid |\log x|>\log \lambda(\beta, d)\}$
  $I=(0, \frac{1}{\lambda^c_{\beta, d}})\cup ({\lambda^c_{\beta, d}}, \infty)$ 
   \cite{lsscorrelation}\\
  7 & 2-Spin & None & \bf Yes & Yes & Four regions \cite{shao2019contraction}\\
\hdashline
 8 & Ising & $\beta>1$ or $\beta<-1$ & No & No & $|\lambda|\neq 1$  \cite{LeeYang2}  \\
      9& Ising & $\lambda=1$ & No & Yes & $\{\beta\in \mathbb{C}\mid |\frac{\beta-1}{\beta+1}|<\delta_d\}$ \cite{BarvinokBook,mann2019approximation}\\
   10& Ising & $\lambda=1$ & No & Yes & $\beta\in \mathcal{N}(1)$, diamond shape \cite{BB21}\\
       11& Ising & $\lambda=1$ & No & Yes & $\{\beta\in \mathbb{C}\mid |\frac{\beta-1}{\beta+1}|<\tan\frac{\pi}{4(d-1)}\}$ \cite{GGHP22}\\
12& Ising & $\lambda=1$ & No & Yes & $\{\beta\in \mathbb{C}\mid |\frac{\beta-1}{\beta+1}|<\frac{1-o_{d}(1)}{d-1}\}$ \cite{PRS23}\\
 13 & 2-Spin & $\beta\gamma>1,\beta\ge\gamma$ & No & No & $\lambda\in \mathcal{N}(I)$, $I=[0,(\frac{\beta}{\gamma})^{l(\beta,\gamma)})$ \cite{zerofe} \\
 \hline
\end{tabular}

    \footnotemark[2]{\cite{petersregts20} also gives a zero-freeness result on a part of a circle, which is not a region.}
    
     \caption{Zero-free regions for 2-spin systems} \label{tab:my_label}
\end{table}

For each zero-free region in Table~\ref{tab:my_label}, an FPTAS exists via Barvinok’s algorithm.
In contrast, Weitz-style FPTASes are currently known only for the regions in Rows 1, 2, 4, 6, and 7, where SSM can be established via tree recurrence arguments using contraction properties.
For other regions, particularly those corresponding to the ferromagnetic regime ($\beta\gamma > 1$), it remains unclear whether SSM can be proved using such methods.
This raises a natural and intriguing question: \emph{For those zero-free regions of 2-spin systems where FPTASes exist via Barvinok’s algorithm, can Weitz’s algorithm also be applied? In particular, is it possible to establish SSM for these zero-free regions even when the standard tree recurrence method is not known to work?}

\subsection*{New SSM results and Weitz's FPTASes beyond tree recurrences}

In this paper, we answer this question in the affirmative.
We show that for any zero-free region of 2-spin systems on graphs with pinned vertices,  SSM holds throughout the same region.
This immediately implies new SSM results including the regions in Rows 3 and 5 of Table~\ref{tab:my_label}.   
We also obtain new SSM results for parts of the Lee-Yang zero-free region. 
Consequently, Weitz's FPTASes exist for these regions.
Notably, these SSM results are \emph{not} known  to be provable via the standard tree recurrence method. 

\begin{theorem}\label{thm-intro-new-SSM}
The  SSM property holds for every zero-free region in Table~\ref{tab:my_label}.
  In addition, for the ferromagnetic Ising model ($\beta=\gamma>1$), the  SSM property holds for the regions
   $\{\lambda\in \mathbb{C}\mid |\lambda|<1/\beta^d\}$ and $\{\lambda\in \mathbb{C}\mid |\lambda|>\beta^d\}$. 
As a consequence, Weitz's FPTASes exist for these regions.
\end{theorem}

\begin{remark}

The conditions $|\lambda|<1/\beta^d$ and  $|\lambda|>\beta^d$ correspond to subsets of the zero-free regions $|\lambda|\neq 1$ established by the Lee-Yang circle theorem.
 At first glance,  establishing SSM under these conditions might appear feasible via a standard tree recurrence argument.
However, to the best of our knowledge, \emph{no such tree recurrence proof is known, or even believed feasible}.
The tree recurrence method typically relies on carefully designed potential functions, the construction and analysis of which can be highly challenging  (e.g., see Section 4 of \cite{li2013correlation}). 
For instance, in antiferromagnetic 2-spin systems, such a function has been constructed only within the uniqueness region, exhibiting  a  complicated form expressed as an antiderivative of $\frac{1}{\sqrt{x(\beta x+1)(x+\gamma)}}$.
Given the difficulty in constructing potential functions, whether one can be constructed for ferromagnetic 2-spin systems remains an open question. We overcome this obstacle by presenting a  direct proof that establishes \emph{novel} SSM results without employing the tree recurrence method.

Furthermore, we introduce two variants of spatial mixing, namely \emph{plus spatial mixing} (PSM) and \emph{minus spatial mixing} (MSM) (see Definition~\ref{def:minus-mixing}) 
corresponding to Ising models with all-plus and all-minus boundary conditions, respectively.
We  establish them for the ferromagnetic Ising model in  broader parameter regimes  $\lambda>\beta$ and $\lambda<1/\beta$, respectively. 
\end{remark}

\subsection*{Related work}
Our work generalizes a recent line of work establishing  that zero-freeness implies SSM for the hard-core model and other graph homomorphism models~\cite{harrow2020classical,gamarnik2020correlation, Guus2021zerofreetossm}. 
This implication reflects a deeper connection between the two key paradigms for designing FPTASes for 2-spin systems: zero-freeness (from complex analysis) and correlation decay (from probabilistic arguments).
 In seminal works \cite{dobrushin1985completely, dobrushin1987completely}, Dobrushin and Shlosman proved that for lattice models, the correlation decay property is equivalent to the zero-freeness of the partition function. 
Recently,  further formal connections have been established beyond lattice models.
Such results have been obtained for the hard-core model~\cite{PetersRegts19}, the Ising model~\cite{LSSFisherzeros, petersregts20, lsscorrelation}, and  general 2-spin systems~\cite{shao2019contraction} on graphs with bounded degree. 
In these works, for real parameters where correlation decay can be established via the tree recurrence method using the contraction property, this property can often be analytically extended to complex neighborhoods of the parameters, thereby ensuring zero-freeness of the partition function.
However, it remains unclear whether correlation decay or SSM directly implies zero-freeness of the partition function.
On the other hand, \cite{harvey2018computing} proved correlation decay result for hard-core model based on multivariate zero-freeness condition on polydisc region, establishing Weitz's type FPTAS. 
Gamarnik \cite{gamarnik2020correlation} first showed that zero-freeness of the partition function
directly implies a weak form of SSM for the hard-core model and  other
graph homomorphism models. 
Later, Regts~\cite{Guus2021zerofreetossm} strengthened this result by proving that zero-freeness implies  SSM for the hard-core model near $\lambda = 0$, as well as for other graph homomorphism models on all graphs of bounded degree.
When restricted to  2-spin systems, the graph homomorphism model is
the 2-spin system without external field, i.e., $\lambda=1$, and 
the implication from zero-freeness to SSM holds for a complex neighborhood of  $(\beta, \gamma)=(1, 1)$. 

\section*{A novel and unified approach beyond cluster expansions}

In this paper, we significantly extend the known implication from zero-freeness to the most general settings of 2-spin systems, comprehensively covering all currently known zero-free regions (Rows 1 to 7 of  Table~\ref{tab:my_label}, to the best of our knowledge) of the partition functions of  2-spin systems with pinned vertices. 
The prior approach \cite{Guus2021zerofreetossm} relies on  an analytic tool, namely cluster expansions, which are model-specific and known only for a few cases such as the hard-core model (i.e., $\beta=0, \gamma=1$) near $\lambda=1$, or the 2-spin system without external field (i.e., $\lambda=1$) near $\beta=1$ and $\gamma=1$.
A crucial limitation of this approach is that
beyond these model-specific  settings, general cluster expansions are not known for  2-spin systems with arbitrary complex parameters $(\beta, \gamma, \lambda)$ on general graphs.\footnote{While there are cluster expansion results  for lattice structures~\cite{friedli2018statistical}, it is unknown whether these results can be extended to general graphs.} 

We circumvent this obstacle and establish SSM for general 2-spin systems via a purely combinatorial approach based on a \emph{Christoffel–Darboux-type identity} (Theorem~\ref{CDgeneralintro} below), which holds universally for  arbitrary complex $(\beta, \gamma, \lambda)$. This enables a unified analysis of various  zero-free regions  across all three parameters.
 Unlike the previous result~\cite{Guus2021zerofreetossm}, which treats the partition function as a univariate function restricted to strip-like zero-free regions, 
 we are able to handle the partition function $Z_G(\beta, \gamma, \lambda)$ with arbitrary \emph{multivariate} {complex} parameters $(\beta, \gamma, \lambda)$ and zero-free regions of arbitrary shape.
In summary, we present a unified and broadly applicable framework for deriving SSM from zero-freeness based on a combinatorial identity, bypassing the need for model-specific cluster expansions which may be inapplicable to general 2-spin systems. 
In Appendix~\ref{app:comparison}, we review the cluster expansion approach in more detail, discuss its limitations in extending to general 2-spin systems, and compare it with our approach.
%\footnote{In a recent follow-up work~\cite{shao2025zero}, this framework has been further extended to more general settings, including hypergraph independence polynomials, random cluster models, Potts models, and Holant problems, thereby yielding new SSM results.} 

Below, we formally describe the implication from zero-freeness to SSM for general 2-spin systems.
A region refers to a standard complex region, i.e., a simply connected open set in $\mathbb{C}$, and a neighborhood of some $z\in \mathbb{C}$ refers to  a region containing $z$. 
For $z\in \mathbb{C}$, we say $U\subseteq \mathbb{C}$ is a  quasi-neighborhood of $z$ if either $U$ is a  neighborhood of $z$ or $U=\{z\}$. 
For a region $U$ where $0\notin U$, we define $1/U$ to be the set $\{z\in \mathbb{C}\mid 1/z \in U\}$,
which is a region. 
A family $\mathcal{G}$ of graphs is closed under constructions of \emph{self-avoiding walk trees} (see Subsection \ref{sec:saw-tree})  if for any graph $G\in \mathcal{G}$ and any vertex $v$ of $G$, the SAW tree  of $G$ rooted at $v$ is in $\mathcal{G}$.

\begin{theorem}\label{thm:main}
    Let $\mathcal{G}$ be
a family of graphs closed under self-avoiding walk {\rm (SAW)} tree construction, and $U_1, U_2, U_3\subseteq \mathbb{C}$ be one of the following cases:
\begin{enumerate}\setlength{\itemsep}{-0.5ex}
    \item $U_1$ and $U_2$ are quasi-neighborhoods of some $\beta_0, \gamma_0{\geq 0}$ respectively where $\beta_0, \gamma_0$ are not both zero, and 
$U_3$ is a neighborhood of $0$;
 \item $U_2$ and $U_3$ are quasi-neighborhoods of some $\gamma_0, \lambda_0 {>0}$ respectively where $0\notin U_2$, and $U_1$ is a region containing $1/U_2$.

\item $U_1$ and $U_3$ are quasi-neighborhoods of some $\beta_0, \lambda_0 {>0}$ respectively where $0\notin U_1$, and $U_2$ is a region containing $1/U_1$. 

\end{enumerate}
Denote  $((U_1\times U_2) \backslash\{(0, 0)\})\times (U_3\backslash\{0\})$ by $\mathbf U$. 
Suppose for any graph $G\in \mathcal{G}$ and any feasible partial configuration $\sigma_\Lambda$ (Definition~\ref{def:feasible-configuration}), the partition function $Z^{\sigma_\Lambda}_{G}(\beta, \gamma, \lambda)$ is zero-free on $\mathbf U$.
Then, for any $(\beta, \gamma, \lambda) \in \mathbf U$, the 2-spin system defined on $\mathcal{G}$ with parameters $(\beta, \gamma, \lambda)$ exhibits SSM (Definition~\ref{def:ssm}). 

Moreover, when $\mathcal{G}$ is
 a family of graphs of bounded degree, the parameters $\beta_0, \gamma_0$ and $\lambda_0$ can be relaxed to any nonzero complex numbers satisfying $\beta_0\gamma_0=1$.
\end{theorem}

Roughly speaking, the first condition of Theorem~\ref{thm:main} requires that the partition function viewed as a function of $\lambda$ is zero-free in a complex neighborhood of $\lambda=0$, and the second and third conditions of Theorem~\ref{thm:main} together require that the partition function viewed as a function of $\beta$ or $\gamma$ is zero-free in a neighborhood of $\beta\gamma =1$.
We use Figure~\ref{fig:main} to illustrate Theorem~\ref{thm:main}. The first row and the second row are examples of conditions 1 and 2 respectively, where the underlying graphs are arbitrary and a non-negative point $(\beta_0, \gamma_0, \lambda_0)$ is required. The third row illustrates condition 3, where the underlying graphs are required to have bounded degree, but no non-negative point is required. 

\begin{figure}[!hbtp]\label{fig:main}
\centering
\hspace{-5ex}	\includegraphics[scale=0.4]{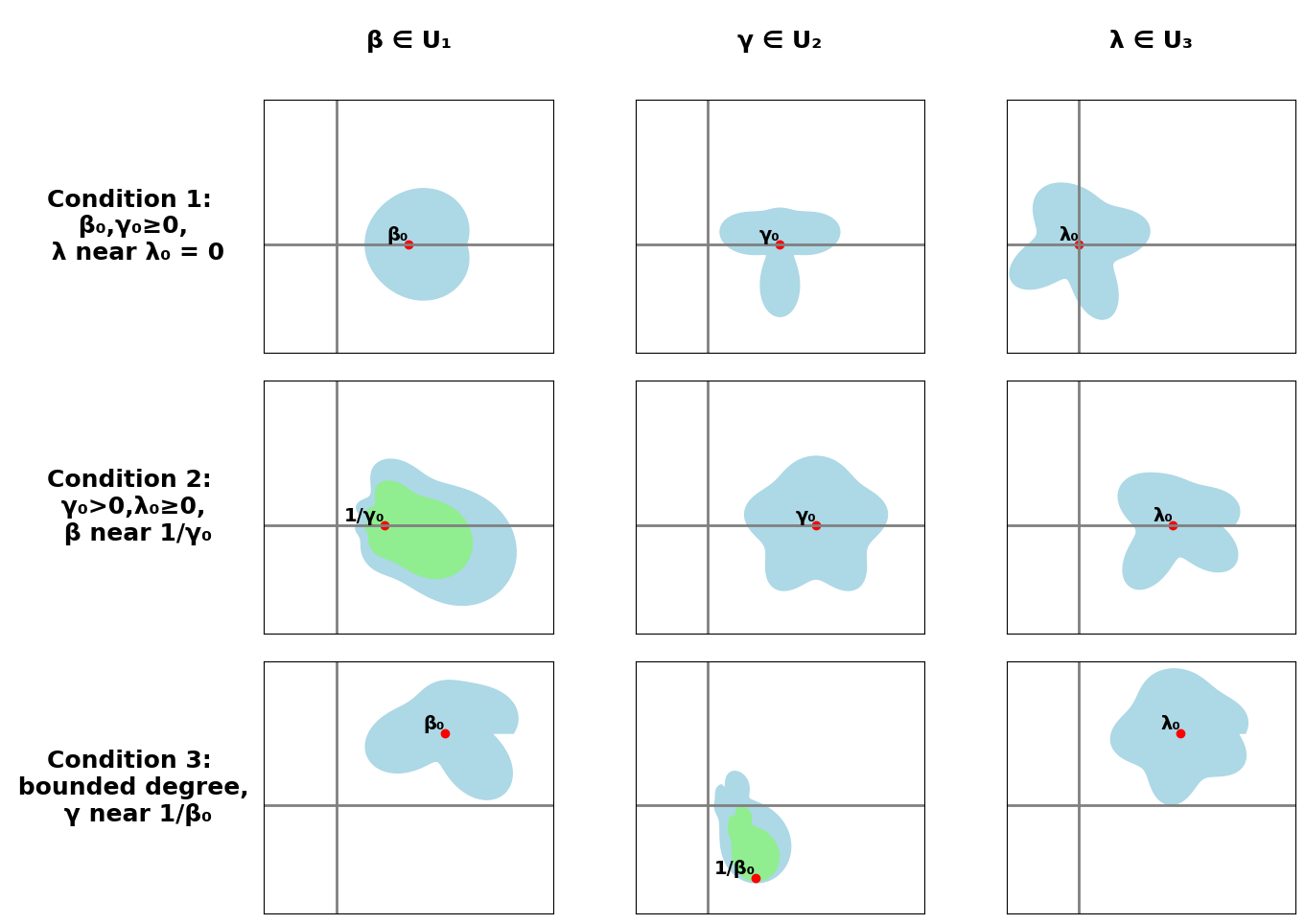}
	\caption{Examples of regions for the 3 conditions in Theorem \ref{thm:main}}
	\label{fig:condition-main}
	\end{figure}

In addition, for the Ising model, we have the following SSM result from zero-free regions of $\beta$. 
\begin{theorem}\label{mainIsing}
    Let $\mathcal{G}$ be
a family of graphs closed under {\rm SAW} tree constructions,  $\lambda_0$ be a nonzero complex number, and $U$ be a neighborhood of $1$ or $-1$. 
Suppose that for any graph $G\in \mathcal{G}$ and any feasible partial configuration $\sigma_\Lambda$, the partition function $Z^{\sigma_\Lambda}_{G,\lambda_0}(\beta)=Z^{\sigma_\Lambda}_{G}(\beta, \beta, \lambda_0)$ is zero-free on $\beta \in  U$.
Then, for any $\beta \in  U$, the corresponding Ising model specified by $\beta$ and $\lambda_0$ exhibits SSM. 
\end{theorem}

Although we give three different conditions involving three parameters $\lambda, \beta$, and $\gamma$ for 2-spin systems (Theorem~\ref{thm:main}) and a condition for Ising models (Theorem~\ref{mainIsing}) under which zero-freeness implies SSM, they are all proved by a unifying approach in which a Christoffel-Darboux type identity plays a key role.
A Christoffel-Darboux type identity was originally established for the independence polynomial on trees~\cite{Gutman} and was later extended to general graphs~\cite{Bencs18}.
Our main technical contribution lies in
establishing a Christoffel-Darboux type identity for the 2-spin system\footnote{In fact, we establish an identity for the more general $q$-spin system in the appendix, see Theorem \ref{thm:CD-q-spin}.} on trees.
Such an identity provides a relation for the partition functions of 2-spin systems with different pinned vertices, and may be of independent interest. 
Given a graph $G$, a partial configuration $\sigma_{\Lambda}$, and  $u,v\in V(G),$ we define
\[Z^{\sigma_\Lambda,+}_{G,v}(\beta,\gamma,\lambda)=\sum\limits_{\sigma:V\to\{+,-\} \atop \sigma|_{\Lambda}=\sigma_\Lambda, \sigma(v)=+}w(\sigma), \text{ and } Z^{\sigma_\Lambda,+,+}_{G,v,u}(\beta,\gamma,\lambda)=\sum\limits_{\sigma:V\to\{+,-\} \atop  \sigma|_{\Lambda}=\sigma_\Lambda,  \sigma(v)=\sigma(u)=+}w(\sigma),\]
and similarly we can define $Z^{\sigma_\Lambda,-}_{G,v}(\beta,\gamma,\lambda),Z^{\sigma_\Lambda,+,-}_{G,v,u}(\beta,\gamma,\lambda),Z^{\sigma_\Lambda,-,+}_{G,v,u}(\beta,\gamma,\lambda)$ and $Z^{\sigma_\Lambda,-,-}_{G,v,u}(\beta,\gamma,\lambda).$

\begin{theorem}\label{CDgeneralintro}
    Suppose that $T$ is a tree, $\sigma_{\Lambda}$ is a partial configuration on some $\Lambda \subseteq V$ which may be empty, 
    and $u$ and $v$ are two distinct vertices in $V\backslash \Lambda$. 
    Let $d(u,v)$ denote the distance between $u$ and $v$, $p_{uv}$ denote  the unique path in $T$ connecting $u$ and $v$, $V(p_{uv})$ denote the set of vertices in $p_{uv}$, $N[p_{uv}]$ denote the set of neighbors of vertices in $p_{uv}$,  $N[p_{uv}]\setminus V(p_{uv})=\{v_1,\cdots,v_n\}$, and $T_i$ denote the component of $v_i$ in $T\setminus p_{uv}$. 
    Then for the partition function of the 2-spin system defined on $T$ conditioning on $\sigma_{\Lambda}$, we have {(we omit the argument $(\beta, \gamma, \lambda)$ in the second line for simplicity)}
\begin{equation}\label{equ-cd}
\begin{aligned}
& Z_{T,u,v}^{\sigma_{\Lambda},+,+}(\beta,\gamma,\lambda)Z_{T,u,v}^{\sigma_{\Lambda}, -,-}(\beta,\gamma,\lambda) 
-Z_{T,u,v}^{\sigma_{\Lambda}, +,-}(\beta,\gamma,\lambda)Z_{T,u,v}^{\sigma_{\Lambda}, -,+}(\beta,\gamma,\lambda) \\
=&\begin{cases}(\beta\gamma-1)^{d(u,v)}\lambda^{d(u,v)+1}\prod\limits_{i=1}^{n}(\beta Z_{T_i,v_i}^{\sigma_{\Lambda},+}+Z_{T_i,v_i}^{\sigma_{\Lambda},-})(Z_{T_i,v_i}^{\sigma_{\Lambda},+} +\gamma Z_{T_i,v_i}^{\sigma_{\Lambda},-}) & , \text{ if } V(p_{uv})\cap\Lambda=\emptyset \\
0 & , \text{ if } V(p_{uv})\cap \Lambda\neq\emptyset
\end{cases}.\end{aligned}
\end{equation}    
\end{theorem}

Although Theorem~\ref{CDgeneralintro} is established for trees only,  the implication from zero-freeness to SSM applies to 2-spin systems on all graphs by using SAW trees (Lemma~\ref{lem:tree-enough}).
Thus, for our purposes, it suffices to  prove the identity  for trees only.
It remains an interesting and challenging question whether this identity can be extended to 2-spin systems on general graphs.

 We briefly explain why Theorem~\ref{CDgeneralintro} provides a unifying approach to turn all existing zero-free regions of  2-spin systems to SSM. 
A key step in deriving SSM from zero-freeness is to show that the  coefficients of the first $k$ terms in the Taylor series of the rational function $P_{G,v}^{\sigma_\Lambda}(\beta,\gamma,\lambda)=\frac{Z^{\sigma_\Lambda,+}_{G,v}(\beta,\gamma,\lambda)}{Z^{\sigma_\Lambda}_{G}(\beta,\gamma,\lambda)}$ around a specific point (for example $\lambda=0$, or $\beta=1$ after a variable replacement) depend only on the $k$-neighborhood of $v$ in $G$, defined as LDC (Definition~\ref{def:LDC}).
By Theorem~\ref{CDgeneralintro}, one can factor out $\lambda^{d(u,v)}$ and $(\beta\gamma-1)^{d(u,v)}$
from $P_{G,v}^{\sigma_\Lambda}(\beta,\gamma,\lambda)-P_{G,v}^{\sigma_\Lambda, u^+}(\beta,\gamma,\lambda)$ where $P_{G,v}^{\sigma_\Lambda, u^+}(\beta,\gamma,\lambda)=\frac{Z^{\sigma_\Lambda,+,+}_{G,v,u}(\beta,\gamma,\lambda)}{Z^{\sigma_\Lambda,+}_{G,u}(\beta,\gamma,\lambda)}$.
Then, by pinning vertices pointwise, we can eventually prove LDC for 2-spin systems both near $\lambda=0$ and $\beta\gamma=1$ in a unified way, eliminating the need to establish cluster expansions for various special cases of 2-spin systems to prove LDC.
Thus, the use of Theorem~\ref{CDgeneralintro} is not only novel but also circumvents the need for cluster expansion, which may be inapplicable to general 2-spin systems. 

Furthermore, the application of Theorem~\ref{CDgeneralintro} reveals another potential advantage. It may answer a natural and intriguing question: \emph{if the partition function is zero-free on a complex neighborhood of some points other than $\lambda=0$ and $\beta\gamma=1$, can we still derive SSM from such zero-free regions?} The answer would be yes if one could factor out some $(\lambda-\lambda_0)^{d(u,v)}$ or $(\beta\gamma-z_0)^{d(u,v)}$ from $P_{G,v}^{\sigma_\Lambda}(\beta,\gamma,\lambda)-P_{G,v}^{\sigma_\Lambda, u^+}(\beta,\gamma,\lambda)$. However, such a factorization is ruled out by Theorem~\ref{CDgeneralintro}.
Otherwise, there would exist some $\lambda_0\neq 0$ or $\beta_0\gamma_0=z_0\neq 1$ such that $P_{G,v}^{\sigma_\Lambda}(\beta,\gamma,\lambda_0) \equiv P_{G,v}^{\sigma_\Lambda, u^+}(\beta,\gamma,\lambda_0)$ or  $P_{G,v}^{\sigma_\Lambda}(\beta_0,\gamma_0,\lambda) \equiv P_{G,v}^{\sigma_\Lambda, u^+}(\beta_0,\gamma_0,\lambda)$ for all graphs $G$, which is impossible.  
Thus, it is highly possible that by using Theorem~\ref{CDgeneralintro}, we have indeed characterized  all zero-free regions from which SSM can be derived (at least via the approach based on establishing LDC). 
In addition, we utilize the Riemann mapping theorem to handle zero-free regions of arbitrary shape, rather than restricting to strip-shaped neighborhoods of real intervals.\footnote{In order to obtain an FPTAS for non-disk zero-free regions via Barvinok's algorithm, such regions must be mapped to a disk by polynomials~\cite{BarvinokBook}. While by using Riemann
mapping theorem, one can establish SSM and obtain an FPTAS via Weitz's algorithm for  zero-free regions of arbitrary shapes.}

\subsection*{Organization}
The paper is organized as follows. In Section~\ref{sec:preliminaries}, we present some definitions and notations, and introduce certain results from complex analysis. 
In Section~\ref{sec:main proof}, we prove a Christoffel-Darboux type identity for 2-spin systems on trees, and use it to prove the key LDC property for $\beta, \gamma$ and $\lambda$ in a unified way.  
%We apply the approach to show that zero-freeness implies SSM for the external field $\lambda$. 
In Section~\ref{sec:together}, we combine all these ingredients together and show that zero-freeness implies SSM for regions of arbitrary shapes. 
In Section~\ref{sec:LY}, we extend our approach to deal with 2-spin systems with non-uniform external fields, and establish novel (strong/plus/minus) spatial mixing results for the non-uniform ferromagnetic Ising model from the Lee-Yang circle theorem. 
We conclude the paper and introduce some open questions in Section~\ref{sec:conclude}.
An extension of the Christoffel-Darboux type identity to \texorpdfstring{$q$}{q}-spin systems on trees will be given in Appendix~\ref{sec:CDqspin}.
In Appendix~\ref{app:comparison}, we compare our approach with the cluster expansion approach in more detail.

\section{Preliminaries}\label{sec:preliminaries}

\subsection{Notations and definitions}
Given a partial configuration $\sigma_\Lambda$,  for any $\Lambda'\subset\Lambda,$ we denote by $\sigma_{\Lambda'}$ the restriction of $\sigma_{\Lambda}$ on $\Lambda'$, called a  sub-partial configuration of $\sigma_{\Lambda}$. 
Consider the  rational function $P_{G,v}^{\sigma_\Lambda}(\beta,\gamma,\lambda)=\frac{Z^{\sigma_\Lambda,+}_{G,v}(\beta,\gamma,\lambda)}{Z^{\sigma_\Lambda}_{G}(\beta,\gamma,\lambda)}.$
Note that when $\beta, \gamma, \lambda{>0}$, $P_{G,v}^{\sigma_\Lambda}$ denotes the marginal probability of $v$ being assigned to the spin $+$ in the Gibbs distribution. 
    However, when $\beta, \gamma, \lambda\in \mathbb{C}$, $P_{G,v}^{\sigma_\Lambda}(\beta,\gamma,\lambda)$ has no probabilistic meaning and $P_{G,v}^{\sigma_\Lambda}(\beta,\gamma,\lambda)\in \mathbb{C}$  if $Z^{\sigma_\Lambda}_{G}(\beta,\gamma,\lambda)\neq 0$. 
For the setting $\beta =0$ or $\gamma =0$,
there are trivial partial configurations such that
 $Z^{\sigma_\Lambda}_{G, v}(\beta,\gamma,\lambda)=0$. 
We rule these cases out as they are \emph{infeasible}.

\begin{definition}[Feasible partial configuration]\label{def:feasible-configuration}
Given a  2-spin system defined on a graph $G$ with parameters $(\beta,\gamma,\lambda)$,
a partial configuration $\sigma_\Lambda$ is feasible if 
 $\sigma_\Lambda$ does not assign any two adjacent vertices in $G$ both to the spin $+$ when $\beta=0$, and  $\sigma_\Lambda$ does not assign any two adjacent vertices in $G$ both to the spin $-$ when $\gamma=0$.
 We say a vertex $v\in V$ is proper to a feasible partial configuration $\sigma_\Lambda$ if $v\notin \Lambda$ and the partial configurations $\sigma^+_{\Lambda\cup\{v\}}$  and 
 $\sigma^-_{\Lambda\cup\{v\}}$ 
 obtained from $\sigma_\Lambda$ by further pinning $v$ to $+$ and $-$ respectively are still feasible. 
\end{definition}
\begin{remark}
    Any sub-partial configuration of a feasible partial configuration is still feasible.
\end{remark}

\begin{definition}[Zero-freeness]
    Let $\mathbf{U}\subseteq \mathbb{C}^3$ and $\mathcal{G}$ be a family of graphs. The partition function of the 2-spin systems defined on $\mathcal{G}$ is said to be \emph{zero-free} on $\mathbf{U}$  if for any graph $G\in \mathcal{G}$, any feasible partial configuration $\sigma_\Lambda$, the partition function  $Z^{\sigma_\Lambda}_{G}(\beta,\gamma,\lambda)\neq 0$ for  all $(\beta, \gamma, \lambda)\in \mathbf{U}$. 
\end{definition}

When the  parameters  $\beta$ and $\gamma$ of a 2-spin system are fixed, we may write the partition function $Z^{\sigma_\Lambda}_{G}(\beta,\gamma,\lambda)$ as $Z^{\sigma_\Lambda}_{G,\beta,\gamma}(\lambda)$ and the rational function $P_{G,v}^{\sigma_\Lambda}(\beta,\gamma,\lambda)$ as $P_{G, \beta,\gamma,v}^{\sigma_\Lambda}(\lambda)$, which are both univariate functions on $\lambda$. 
 Sometimes we may omit the fixed parameters $\beta$ and $\gamma$ at the subscript for simplicity.
The definition of zero-freeness can be easily adapted to the partition function $Z^{\sigma_\Lambda}_{G,\beta,\gamma}(\lambda)$ with fixed $\beta$ and $\gamma$, and similarly to $Z^{\sigma_\Lambda}_{G,\gamma,\lambda}(\beta)$ with fixed $\gamma$ and $\lambda$, as well as $Z^{\sigma_\Lambda}_{G,\beta,\lambda}(\gamma)$ with fixed $\beta$ and $\lambda$. 
When the three parameters are all fixed, we may write the value $Z^{\sigma_\Lambda}_{G}(\beta,\gamma,\lambda)$ as $Z^{\sigma_\Lambda}_{G}$ and the value $P_{G,v}^{\sigma_\Lambda}(\beta,\gamma,\lambda)$ as $P_{G,v}^{\sigma_\Lambda}$ for simplicity. 

Suppose that the partition function $Z^{\sigma_\Lambda}_{G,\beta,\gamma}(\lambda)$ is zero-free on some $U\subseteq \mathbb{C}$.
Then, the rational function $P_{G, \beta,\gamma,v}^{\sigma_\Lambda}(\lambda)$ is analytic on $U$ for any graph $G$ and any feasible partial configuration $\sigma_\Lambda$.
Moreover, if the vertex $v$ is proper to $\sigma_\Lambda$, then $P_{G, \beta,\gamma,v}^{\sigma_\Lambda}(\lambda) \neq 0$ and $P_{G, \beta,\gamma,v}^{\sigma_\Lambda}(\lambda) \neq 1$ on $U$ since $Z_{G, \beta,\gamma,v}^{\sigma_\Lambda, +}(\lambda)\neq 0$ and  $Z_{G, \beta,\gamma,v}^{\sigma_\Lambda, -}(\lambda)\neq 0$. 
Notice that $\lim_{\lambda\rightarrow 0}P_{G, \beta,\gamma,v}^{\sigma_\Lambda}(\lambda)=0$.
We agree that $P_{G, \beta,\gamma,v}^{\sigma_\Lambda}(\lambda)=0$ when $\lambda=0$. 
If the zero-free region $U$ is of the form $\mathcal{N}(0)\backslash\{0\}$ for some complex neighborhood $\mathcal{N}(0)$ of $0$, then the rational function $P_{G, \beta,\gamma,v}^{\sigma_\Lambda}(\lambda)$ is analytic on the entire $\mathcal{N}(0)$. 

\begin{definition}[Strong spatial mixing, SSM]\label{def:ssm}
    Fix complex parameters $\beta, \gamma, \lambda$ where $(\beta, \gamma)\neq (0, 0)$ and $\lambda\neq 0$, and a family of graphs $\mathcal{G}$. 
The corresponding 2-spin system defined on  $\mathcal{G}$ with parameters $(\beta, \gamma, \lambda)$ is said to satisfy \emph{strong spatial mixing (SSM)} with exponential rate $r>1$ if there exists a constant $C$ such that for any $G=(V,E)\in\mathcal{G}$,
any feasible partial configurations $\sigma_{\Lambda_1}$ and $\tau_{\Lambda_2}$ where $\Lambda_1$ may differ from $\Lambda_2$, 
and any vertex $v$ proper to $\sigma_{\Lambda_1}$ and $\tau_{\Lambda_2}$, we have 
\[\left|P_{G,v}^{\sigma_{\Lambda_1}}-P_{G,v}^{\tau_{\Lambda_2}}\right|\leq Cr^{-d_G(v,\sigma_{\Lambda_1}\neq \tau_{\Lambda_2})}.\]
Here,  we denote by $\sigma_{\Lambda_1}\neq \tau_{\Lambda_2}$ the set $(\Lambda_1\setminus\Lambda_2)\cup (\Lambda_2\setminus\Lambda_1)\cup \{v\in\Lambda_1\cap\Lambda_2:\sigma_{\Lambda_1}(v)\neq\tau_{\Lambda_2}(v)\}$ (i.e., the set on which $\sigma_{\Lambda_1}$ and $\tau_{\Lambda_2}$ differ with each other), and $d_G(v,\sigma_{\Lambda_1}\neq \tau_{\Lambda_2})$ is the shortest path distance from $v$ to any vertex in $\sigma_{\Lambda_1}\neq \tau_{\Lambda_2}$.
    
\end{definition}

    Note that the SSM property guarantees that the partition function $Z_{G}^{\sigma_\Lambda}$ is zero-free on the fixed point $(\beta, \gamma, \lambda)$, as $|P_{G,v}^{\sigma_{\Lambda}}|<\infty$ for any graph $G$ and any feasible partial configuration $\sigma_{\Lambda}$. 
   The above SSM definition for complex parameters is stronger than the usual definition for real parameters (see Definition 5 of \cite{li2013correlation} for example). 
For real values, by monotonicity one  need only consider the case that $\Lambda_1=\Lambda_2$ (the two configurations are on the same set of vertices). 
However when defining SSM for complex parameters, it is necessary to consider the case that $\Lambda_1\neq \Lambda_2$ 
in order to ensure that the SSM-based Weitz's algorithm works for complex parameters (see Lemma 3.4 of \cite{shao2019contraction}).

\subsection{Weitz's algorithm and self-avoiding walk tree}\label{sec:saw-tree}
  Given a graph $G=(V,E)$ and a vertex $v\in V$,
 the SAW tree of $G$ at $v$, denoted by $T_{\text{SAW}}(G,v)$, is a tree with the root $v$ that
enumerates all self-avoiding paths beginning from $v$ in.  
Supplementary vertices closing cycles of $G$ are added as leaves of the tree. 
Each vertex in $V$ of $G$ is mapped to corresponding vertices of $T_{\text{SAW}}(G,v)$. 
A vertex set $\Lambda\subseteq V$ is mapped to  some $\Lambda_{\text{SAW}}\subseteq V_{\text{SAW}}$,
and a partial configuration $\sigma_\Lambda:\Lambda\to\{\pm\}$ is mapped to a partial configuration $\sigma'_{\Lambda_{\text{SAW}}}$ on $\Lambda_{\text{SAW}}$. 
Let $L_{\Lambda_{\text{SAW}}}$ denote the set of leaf vertices $u$ in $V_{\text{SAW}}$ such that $u$ closes a cycle and no vertices in the path between the root $v$ and the leaf $u$ are pinned by $\sigma'_{\Lambda_{\text{SAW}}}$. 
For every vertex in $L_{\Lambda_{\text{SAW}}}$, a configuration is imposed to it depending on the orientation of the cycle, where the order of indices of vertices is arbitrarily chosen (see Figure \ref{fig:saw-tree} for an example).
We denote by $\sigma_{\Lambda_{\text{SAW}}}$ the partial configuration $\sigma'_{\Lambda_{\text{SAW}}}$ on $\Lambda_{\text{SAW}}$ together with the imposed configuration on $L_{\Lambda_{\text{SAW}}}$.
One can easily check that if $\sigma_\Lambda$ is feasible in $G$ and $v$ is proper to $\sigma_\Lambda$, then $\sigma_{\Lambda_{\text{SAW}}}$ is  feasible in $T_{\text{SAW}}(G,v)$  and $v$ is proper to $\sigma_{\Lambda_{\text{SAW}}}$. 
Here is the key result for the SAW tree construction.

\begin{figure}[!hbtp]
\centering
	\includegraphics[scale=0.55]{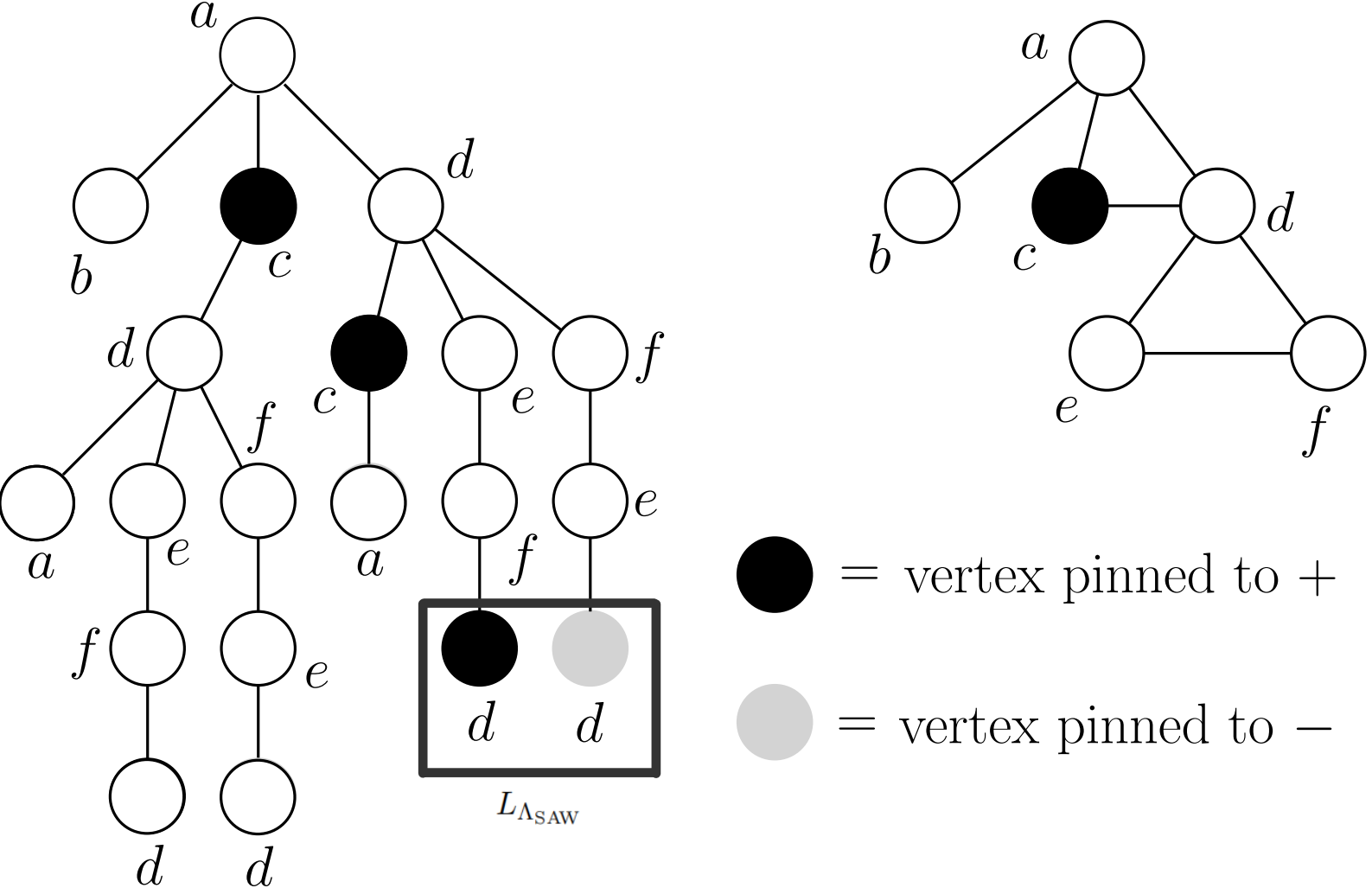}
	\caption{Weitz's construction of SAW tree}
	\label{fig:saw-tree}
	\end{figure}

\begin{theorem}\label{prop:SAW}
  Let $G=(V,E)$, $v\in V$, and $T=T_{\mathrm{SAW}}(G,v)$.
 Suppose that $\sigma_\Lambda$ is a feasible partial configuration on some $\Lambda$ where $v$ is proper to  $\sigma_\Lambda$, and $\sigma_{\Lambda_\mathrm{SAW}}$ is the corresponding feasible partial configuration on $\Lambda_\mathrm{SAW}\subseteq V_\mathrm{SAW}$. Then,
  \begin{enumerate}\setlength{\itemsep}{-0.5ex}
      \item $\mathrm{dist}_G(v,S)=\mathrm{dist}_T(v,S_{\rm{SAW}})$  for any $S\subseteq V$, the $k$-neighborhood $B_G(v,k)$ of $v$ in $G$ is mapped to $B_{T}(v,k)$ in $T$, and the maximum degree of $T$ is equal to that of $G$.
   \item If there exists a region $U\subset \mathbb{C}^3$ such that $Z^{\sigma_{\Lambda_\mathrm{SAW}}}_T(\beta, \gamma, \lambda)\neq 0$ on $U$ (i.e.,  $P_{T, v}^{\sigma_{\Lambda_{\rm{SAW}}}}(\beta, \gamma, \lambda)\neq \infty$), then $Z^{\sigma_\Lambda}_G(\beta, \gamma, \lambda)\neq 0$ and $P_{T, v}^{\sigma_{\Lambda_{\rm{SAW}}}}(\beta, \gamma,  \lambda)=P_{G, v}^{\sigma_\Lambda}(\beta, \gamma,  \lambda)$ on $U$.
  \end{enumerate}
\end{theorem}
\begin{proof}
The first property was proved by Weitz (Theorem 3.1 in \cite{Weitz06}).

The second property was first proved  for the hard-core model with a positive external field $\lambda$ (i.e., $\beta=0$, $\gamma=1$ and $\lambda>0$) by Weitz~\cite{Weitz06}, and
 recently proved for the Ising model with a positive interaction parameter $\beta$ and  a nonzero complex external field $\lambda$ (i.e., $\beta=\gamma>0$ and $\lambda\in \mathbb{C}\backslash\{0\}$) by Peters and Regts \cite{petersregts20}. 
 The current result  for general 2-spin systems with complex parameters $(\beta, \gamma,  \lambda)$ can be proved directly 
following the proof of Proposition 26 in \cite{petersregts20}.
\end{proof}

\subsection{Tools from complex analysis}

For some $\rho>0$ and $z_0\in \mathbb{C}$, we   denote the open disk  $\{z\in\C: |z-z_0|<\rho\}$ centered at $z_0$ of radius $\rho$ by $\mathbb{D}_{\rho}(z_0)$, its boundary by $\partial\mathbb{D}_{\rho}(z_0)$, and its closure by $\overline{\mathbb{D}_{\rho}}(z_0)$.
For simplicity, we write $\mathbb{D}_{\rho}(0)$ as $\mathbb{D}_{\rho}$. 
For a holomorphic (or analytic) function $f$ on a complex neighborhood $U\subseteq \mathbb{C}$ of $z_0$, 
$f$ can be expanded as a convergent series $f(z)=\sum_{k\geq 0}a_k (z-z_0)^k$ near $z_0$, 
called the Taylor expansion series of $f$ at $z_0$.
\begin{theorem}[c.f. Theorem 4.4 in \cite{Stein}]\label{taylor}
    Suppose $f$ is holomorphic on an open set $U\subseteq \mathbb{C}$. If $\mathbb{D}$ is a disc centered at $z_0$ whose closure is contained in $U$, then the Taylor series of $f$ at $z_0$ converges uniformly and equals to $f$ on $\mathbb{D}$. 
\end{theorem}

\begin{lemma}[Cauchy integral formula]
 Suppose $f$ is holomorphic on an open set $U\subseteq \mathbb{C}$. If $\mathbb{D}$ is a disc centered at $z_0$ whose closure is contained in $U$,
 then for every $a\in\mathbb{D}$, 
$$f^{(n)}(a)=\frac{n!}{2\pi i}\int_S \frac{f(z)}{(z-a)^{n+1}}\mathrm{d}z,$$
where $S=\partial\mathbb{D}$ is oriented counterclockwise.
\end{lemma}

A sequence of functions converges locally uniformly to $f$ or $\infty$ on $U$ if for every $z\in U$ there exists an open disk $\mathbb{D}_\rho(z)\subseteq U$ centered at $z$ for some $\rho>0$ such that the sequence converges uniformly to $f$ or $\infty$ on  $\mathbb{D}_\rho(z)$.

\begin{definition}[Normal family]\label{def:normal}
Let $U\subseteq\C^n$ be an open set. A family $\mathcal{F}$ of holomorphic functions $f:U\to \C$ is said to be \emph{normal} if for each infinite sequence of functions of $\mathcal{F}$,  either the sequence has a sub-sequence converging locally uniformly to a holomorphic function or the sequence itself converges locally uniformly to $\infty$ on $U$. 
\end{definition}

The next important tool is Montel's theorem. 
For more details and a proof, please refer to \cite{Narasimhan}  for the version of univariate functions and \cite{multicomplex,dovbush2022applications} for multivariate functions.

\begin{theorem}[Montel's Theorem]\label{thm:montel}
Let $U\subseteq\C^n$ be a connected open set and $\mathcal{F}$ be a family of holomorphic functions $f:U\to \C$ such that $f(U)\subseteq \C \setminus \{0,1\}$ for all $f\in \mathcal{F}$. Then $\mathcal{F}$ is normal.
\end{theorem}

The following Riemann mapping theorem is the key tool to map any bounded complex region to the unit open disk.

\begin{theorem}[Riemann mapping theorem]\label{riemann}
Let $U$ be a simply connected complex neighborhood of some $z_0\in \mathbb{C}$.
There exists a unique bijective holomorphic  function, also called a conformal map $f:\mathbb{D}_1(z_0)\to U$ such that $f(z_0)=z_0$ and $f'(z_0)>0$. The inverse of $f$ is also holomorphic. 
\end{theorem}

%\subsection{Regts' approach}
 %In , %Regts gives a framework , in which 
 The following lemma is used in \cite{Guus2021zerofreetossm} 
 to derive SSM from zero-freeness.

 \begin{lemma}\label{lem:bound strip}
Let $P(z)$ and $Q(z)$ be two analytic functions on some complex neighborhood $U$ of $z_0$. 
Suppose that the Taylor series $\sum_{k\geq 0} a_k (z-z_0)^k$ of $P(z)$ and $\sum_{k\geq 0} b_k (z-z_0)^k$ of $Q(z)$ near $z_0$ 
satisfy $a_k=b_k$ for $k=0, 1, \cdots, n$ for some $n\in \mathbb{N}$.
Also, suppose that there exists an $M>0$ such that both $|P(z)|\le M$ and $|Q(z)|\le M$ on some circle $\partial\mathbb{D}_{\rho}(z_0)\subseteq U$ $(\rho>0)$.
Then for every $z\in \mathbb{D}_{\rho}(z_0)$, we have 
\[\label{eq:bound 1}
    |P(z)-Q(z)|\leq \frac{2M}{\rho(r-1)r^{n}}, \quad \text{with}\quad r=\dfrac{\rho}{|z-z_0|}.
\]
\end{lemma}
\begin{proof}
    Using Cauchy integral formula, we have
    \[\begin{aligned}
    a_k =\dfrac{1}{2\pi i}\int_{\partial\mathbb{D}_{\rho}(z_0)}\dfrac{P(w)}{(w-z_0)^{k+1}}\mathrm{d}w, \text{  \quad and  \quad}
    b_k =\dfrac{1}{2\pi i}\int_{\partial\mathbb{D}_{\rho}(z_0)}\dfrac{Q(w)}{(w-z_0)^{k+1}}\mathrm{d}w. 
    \end{aligned}\]
Note that the series $\sum_{k\geq 0} a_k (z-z_0)^k$ converges to $P(z)$ and $\sum_{k\geq 0} b_k (z-z_0)^k$ converges to $Q(z)$ on the closed disk  $\overline{\mathbb{D}_\rho}(z_0)$ by Theorem \ref{taylor}.
    Thus, for every $z\in\mathbb{D}_\rho(z_0),$ we have
    \begin{equation*}
        \begin{aligned}
        |P(z)-Q(z)| & = \left|\sum\limits_{k\ge n+1}(a_k-b_k)(z-z_0)^k\right|\\
        & \le\sum\limits_{k\ge n+1}\dfrac{|z-z_0|^k}{2\pi} \int_{\partial\mathbb{D}_\rho(z_0)}\left|\dfrac{P(w)-Q(w)}{(w-z_0)^{k+1}}\right|\mathrm{d}w \\
        & \le \dfrac{2M}{\rho}\sum\limits_{k\ge n+1}r^{-k} \\
        & = \frac{2M}{\rho(r-1)r^{n}} 
    \end{aligned}
    \end{equation*}
    which completes the proof. 
\end{proof}
With the lemma in hand, in order to obtain SSM, it suffices to establish the following two properties, namely a uniform bound on a circle and 
local dependence of coefficients (LDC).
  The first property says that there exists a circle $\partial \mathbb{D}_\rho$ around $0$ such that the rational functions $P^{\sigma_{\Lambda}}_{G,v}(\lambda)$ have a uniform bound on $\partial \mathbb{D}_\rho$ for any graph $G$, 
  any feasible partial configuration $\sigma_\Lambda$, and any $v$ proper to it. 
  The second property states that if partial configurations $\sigma_{\Lambda_1}$ and $\tau_{\Lambda_2}$ share the same local structure around a vertex $v$ in $G$ which is proper to both configurations, then the leading Taylor expansion coefficients of the two rational functions $P^{\sigma_{\Lambda_1}}_{G,v}(\lambda)$ and $P^{\tau_{\Lambda_2}}_{G,v}(\lambda)$ near $0$ coincide. 
 %and any $\lambda\in \partial \mathbb{D}_\rho$. 

 In \cite{Guus2021zerofreetossm},
a uniform bound was obtained by applying Montel's theorem (Theorem \ref{thm:montel}).  
In this paper, we use the same idea and extend it to the multivariate version since
 we need to handle three complex parameters $\beta, \gamma$ and $\lambda$.
 %we have to use the multivariate version of Montel's theorem, 
 Also in \cite{Guus2021zerofreetossm}, %Regts' approach, 
%In \cite{Guus2021zerofreetossm}, 
LDC was proved by using the cluster expansion formula for polymer models, since explicit expressions were given for the Taylor series of the rational function $P^{\sigma_\Lambda}_{G,v}(\lambda)$ near $\lambda=0$ where $\beta=0, \gamma=1$, and the rational function $P^{\sigma_\Lambda}_{G,v}(\beta)$ near $\beta=1$ (after a variable replacement) where $\beta=\gamma, \lambda=1$.
%and hence the LDC property was proved. 
However, besides these particular parameters, to the best of our knowledge, explicit cluster expansions are not known for  2-spin systems on general graphs with arbitrary complex parameters $(\beta, \gamma, \lambda)$. 
%However, the 2-spin system on general graphs is not known to be expressible as polymer models except for the case that $\beta=1$ or $\gamma=1$, which  seems to be an intrinsic obstacle for applying cluster expansion to the 2-spin system.
In this paper, we overcome this obstacle and
prove LDC for general 2-spin systems with arbitrary $(\beta, \gamma, \lambda)$ by establishing a Christoffel-Darboux type identity.
It turns out that an explicit expression for the Taylor series of the rational function is unnecessary for proving LDC for 2-spin systems.

 %for $\lambda$ near $0$ with fixed $\beta=0$ and $\gamma=1$ (the hard-core model). 

\section{Christoffel-Darboux type identity gives LDC}\label{sec:main proof}

In this section, we prove our main technical result, a Christoffel-Darboux type identity for 2-spin systems on trees and use it to prove local dependence of coefficients (LDC) for $\lambda, \beta$ and $\gamma$ in a unified way. 
We first formally define LDC.

%\subsection{Local dependence of coefficients and a uniform bound on a circle}\label{sec:ind}

\begin{definition}[Local dependence of  coefficients (LDC) for $\lambda$ near 0]\label{def:LDC}

Fix complex parameters $\beta$ and $\gamma$, a complex neighborhood $U$ of $0$, and a family of graphs $\mathcal{G}$.
The corresponding 2-spin system (defined on $\mathcal{G}$ specified with the fixed parameters $\beta$ and $\gamma$, and a variable $\lambda \in U$) is said to satisfy \emph{local dependence of  coefficients (LDC)} for $\lambda$ near $0$ if for
any $G=(V,E)\in\mathcal{G}$,
any feasible partial configurations $\sigma_{\Lambda_1}$ and $\tau_{\Lambda_2}$, and any vertex $v$ proper to $\sigma_{\Lambda_1}$ and $\tau_{\Lambda_2}$, the functions $P_{G,v}^{\sigma_{\Lambda_1}}(\lambda)$ and $P_{G,v}^{\tau_{\Lambda_2}}(\lambda)$ are analytic on $U$, and their corresponding Taylor series $\sum^\infty_{i=0}a_i\lambda^i$ and $\sum^\infty_{i=0}b_i\lambda^i$ near $\lambda=0$ satisfy $a_i=b_i$ for $0\leq i \leq d_G(v,\sigma_{\Lambda_1}\neq \tau_{\Lambda_2})-1$.  
\end{definition}

\begin{definition}[LDC for $\beta$ near $1/\gamma$]

Fix complex parameters $\gamma\neq 0$ and $\lambda$, a neighborhood $U$ of $1/\gamma$, and a family of graphs $\mathcal{G}$.
The corresponding 2-spin system (defined on $\mathcal{G}$ specified with the fixed parameters $\gamma$ and $\lambda$, and a variable $\beta \in U$) is said to satisfy LDC for $\beta$ near $1/\gamma$ if for
any $G=(V,E)\in\mathcal{G}$,
any feasible configurations $\sigma_{\Lambda_1}\in\{0,1\}^{\Lambda_1}$ and $\tau_{\Lambda_2}\in\{0,1\}^{\Lambda_2}$, and any vertex $v$ proper to $\sigma_{\Lambda_1}$ and $\tau_{\Lambda_2}$, the functions $P_{G,v}^{\sigma_{\Lambda_1}}(\beta)$ and $P_{G,v}^{\tau_{\Lambda_2}}(\beta)$ are analytic on $U$, and their corresponding Taylor series $\sum^\infty_{i=0}a_i(\beta-1/\gamma)^i$ and $\sum^\infty_{i=0}b_i(\beta-1/\gamma)^i$ near $\beta=1/\gamma$ satisfy $a_i=b_i$ for $0\leq i \leq d_G(v,\sigma_{\Lambda_1}\neq \tau_{\Lambda_2})-1$.
\end{definition}

\begin{remark}
Symmetrically, one can define LDC for $\gamma$ near $1/\beta$.     
\end{remark}

As an illustration,  we will focus on the proof of LDC for $\lambda$ near $0$. 
The proofs of LDC for $\beta$ and $\gamma$ are quite similar. 
Below, without other specifications, we use LDC to stand for LDC for $\lambda$ near $0$. 
 We first introduce a special form of  LDC, called point-to-point LDC and we show it implies LDC. 

 \begin{definition}[Point-to-point LDC]\label{def:ppldc}
 Fix complex parameters $\beta$ and $\gamma$, a neighborhood $U$ of $0$, and a family of graphs $\mathcal{G}$. 
The corresponding 2-spin system  is said to satisfy \emph{point-to-point LDC} on $U$ (for $\lambda$ near $0$) if for
any $G=(V,E)\in\mathcal{G}$, any feasible configuration $\sigma_{\Lambda}\in\{0,1\}^{\Lambda}$ on some $\Lambda\subseteq V$, and any two vertices $u,v$ proper to $\sigma_\Lambda$,  the functions  
$P_{G,v}^{\sigma_\Lambda}(\lambda)$,
$P_{G,v}^{\sigma_\Lambda, u^+}(\lambda):=\frac{Z_{G,v,u}^{\sigma_\Lambda,+,+}(\lambda)}{Z_{G,u}^{\sigma_\Lambda,+}(\lambda)}$,  and $P_{G,v}^{\sigma_\Lambda, u^-}(\lambda):=\frac{Z_{G,v,u}^{\sigma_\Lambda,+,-}(\lambda)}{Z_{G,u}^{\sigma_\Lambda,-}(\lambda)}$ are analytic on $U$, and 
their corresponding Taylor series $\sum^\infty_{i=0}a_i\lambda^i$,
 $\sum^\infty_{i=0}b_i\lambda^i$, and $\sum^\infty_{i=0}c_i\lambda^i$ near $\lambda=0$ satisfy $a_i=b_i=c_i$ for $0\leq i \leq d(v, u)-1$. 
\end{definition}

\begin{lemma}\label{PLDCtoLDC}
    If a 2-spin system satisfies point-to-point LDC, then it also satisfies LDC. 
\end{lemma}

\begin{proof}
    When point-to-point LDC holds, we prove LDC by induction on $|\{\sigma_{\Lambda_1}\neq\tau_{\Lambda_2}\}|.$

    When $|\{\sigma_{\Lambda_1}\neq\tau_{\Lambda_2}\}|=0,$ the assertion holds immediately. Suppose now that the assertion holds when $|\{\sigma_{\Lambda_1}\neq\tau_{\Lambda_2}\}|=k\ge0,$ then when $|\{\sigma_{\Lambda_1}\neq\tau_{\Lambda_2}\}|=k+1$, take $u\in\{\sigma_{\Lambda_1}\neq\tau_{\Lambda_2}\},$ by point-to-point LDC, the functions  
$P_{G,v}^{\sigma_{\Lambda_1}}(\lambda)$ and
$P_{G,v}^{\sigma_{\Lambda_1\setminus\{u\}}}(\lambda)$ are analytic on $U$, and 
their corresponding Taylor series $\sum^\infty_{i=0}a_i\lambda^i$ and
 $\sum^\infty_{i=0}b_i\lambda^i$ near $\lambda=0$ satisfy $a_i=b_i$ for $0\leq i \leq d(v,u)-1$; similar equality holds for $\tau_{\Lambda_2}$ and
 $\tau_{\Lambda_2\setminus\{u\}}$. 

 Let $\sigma'=\sigma_{\Lambda_1\setminus\{u\}},\tau'=\tau_{\Lambda_2\setminus\{u\}}$ (recall that $\sigma_{\Lambda_1\setminus\{u\}}$ means the restriction of $\sigma_{\Lambda_1}$ on $\Lambda_1\setminus\{u\}$, if $u\notin \Lambda_1$ then the two partial configurations are the same. ), and then we have $|\{\sigma'\neq\tau'\}|=k$, and hence we can apply the induction hypothesis. Then the conclusion follows since $d_G(v,u)\ge d_G(v,\sigma_{\Lambda_1}\neq \tau_{\Lambda_2})$ and $d_G(v,\sigma'\neq \tau')\ge d_G(v,\sigma_{\Lambda_1}\neq \tau_{\Lambda_2}).$
\end{proof}

\begin{lemma}\label{lem:tree-enough}
    Let $\mathcal{G}$ be a set of graphs, and $\mathcal{T}_{\mathcal{G}}=\{T_{\rm{SAW}}(G,v)\mid G\in\mathcal{G}, v\in V(G)\}$ be the set of all possible {\rm SAW} trees of graphs in $\mathcal{G}$.
    Fix complex parameters $\beta$ and $\gamma$,  and a neighborhood $U$ of $0$.
    If the 2-spin system defined on $\mathcal{T}_{\mathcal{G}}$ satisfies point-to-point LDC, then the 2-spin system defined on $\mathcal{G}$ also satisfies point-to-point LDC.
    %then the 2-spin system defined on ${\mathcal{G}}$ satisfies point-to-point LDC.
\end{lemma}
\begin{proof}
    Given
a graph $G=(V,E)\in\mathcal{G}$, a feasible configuration $\sigma_{\Lambda}\in\{0,1\}^{\Lambda}$ on some $\Lambda\subseteq V$ which may be empty, and two vertices $u,v$ proper to $\sigma_\Lambda$, let $T=T_{\rm SAW}(G,v)$ and $(u)$ denote the set of copies of $u$ in $T$. 
By Theorem~\ref{prop:SAW}, 
we have  
$P_{G,v}^{\sigma_\Lambda}(\lambda)=P_{T, v}^{\sigma_{\Lambda_{\rm{SAW}}}}(\lambda)$,
$P_{G,v}^{\sigma_\Lambda, u^+}(\lambda)=P_{T, v}^{\sigma_{\Lambda_{\rm{SAW}},(u)^+}}(\lambda)$,  and $P_{G,v}^{\sigma_\Lambda, u^-}(\lambda)=P_{T, v}^{\sigma_{\Lambda_{\rm{SAW},(u)^-}}}(\lambda)$, where $\sigma_{\Lambda_{\rm SAW},(u)^+}$ (resp. $\sigma_{\Lambda_{\rm SAW},(u)^-}$) denotes the partial configuration on $T$ obtained by maintaining $\sigma_{\Lambda_{\rm SAW}}$ and pinning all free vertices in $(u)$ to $+$ (resp. $-$).  
These rational functions are analytic on $U$. 

Moreover, notice that $\sigma_{\Lambda_{\rm{SAW}}},\sigma_{\Lambda_{\rm{SAW}},(u)^+}$ and $\sigma_{\Lambda_{\rm{SAW}},(u)^-}$ only differ on a finite set $(u)$ of vertices in $T$ where $d_{T}(v,{u})\ge d_{G}(v,u)$.
We may apply the point-to-point LDC property successively to the SAW tree $T$, and obtain that the Taylor series $P_{T, v}^{\sigma_{\Lambda_{\rm{SAW}}}}(\lambda)=\sum^\infty_{i=0}a_i\lambda^i$,
 $P_{T, v}^{\sigma_{\Lambda_{\rm{SAW}},(u)^+}}(\lambda)=\sum^\infty_{i=0}b_i\lambda^i$, and $P_{T, v}^{\sigma_{\Lambda_{\rm{SAW}},(u)^-}}(\lambda)=\sum^\infty_{i=0}c_i\lambda^i$ near $\lambda=0$ satisfy $a_i=b_i=c_i$ for $0\leq i \leq d(v, u)-1$. 
 Thus, $P_{G,v}^{\sigma_\Lambda}(\lambda)$, $P_{G,v}^{\sigma_\Lambda, u^+}(\lambda)$ and $P_{G,v}^{\sigma_\Lambda, u^-}(\lambda)$ also satisfy point-to-point LDC. 
\end{proof}

Thus, it suffices to prove that point-to-point LDC derives from zero-freeness for trees. 
We prove this by establishing a  Christoffel-Darboux type identity for 2-spin systems on trees. 
%By the SAW tree construction and Theorem~\ref{prop:SAW}, we may restrict the graphs to trees. 

Before we delve into the general 2-spin system, we first consider a simple case, the hard-core model, the partition function of which is the independence polynomial. 
This model is the 2-spin system where $\beta=0$ and $\gamma=1.$ For the hard-core model, \cite{Guus2021zerofreetossm} proved that LDC holds by using a tool called cluster expansion from statistical physics.
Here, we show that it can also be proved by the following Christoffel-Darboux type identity by Gutman \cite{Gutman}. 
For a graph $G=(V, E)$ and a subset $U\subseteq V$, we use $G\backslash U$ to denote the graph obtained from $G$ by deleting all vertices in $U$ and their adjacent edges  (i.e., the subgraph of $G$ induced by the vertex set $V\backslash U$).

\begin{theorem}\label{Gutmanid}
Let $T=(V,E)$ be a tree and $u, v$ be two distinct vertices in $V$. Then, for the partition function of the hard-core model on $G$, we have
    $$Z_T(\lambda) Z_{T\setminus\{u,v\}}(\lambda)-Z_{T\setminus\{u\}}(\lambda)Z_{T\setminus\{v\}}(\lambda)=-(-\lambda)^{d(u,v)+1}Z_{T\setminus p_{uv}}(\lambda)Z_{T\setminus N[p_{uv}]}(\lambda)$$
where $p_{uv}$ is the unique path connecting $u$ and $v$ and $N[p_{uv}]$ is the set of neighbors of points in $p_{uv}$. 
\end{theorem}
   This identity was extended to general graphs in \cite{Bencs18} (Theorem 1.3). 
    By Lemma~\ref{lem:tree-enough},
    %the SAW tree construction and Theorem~\ref{prop:SAW}, 
    the identity for trees is already enough for our proof. 

\begin{lemma}
    For the hard-core model defined on a family $\mathcal{G}$ of trees closed under 
    %{\rm SAW} tree constructions and 
    taking subgraphs, if the partition function $Z_{G}(\lambda)$ is zero-free on a neighborhood $U$ of $0$ for any $G\in \mathcal{G}$, then it satisfies point-to-point LDC on $U$. 
\end{lemma}
\begin{proof}
The case $d_G(u,v)=1$ is easy to check. Here, we assume that $d_G(u,v)\ge 2$. For the hard-core model, $Z^+_{G, v}(\lambda)=\lambda Z_{G\backslash N[v]}(\lambda)$ and $Z^-_{G, v}(\lambda)=Z_{G\backslash v}(\lambda)$. 
Then, $$P^{u^-}_{G, v}(\lambda)=\frac{Z^{-,+}_{G,u,v}(\lambda)}{Z^-_{G,u}(\lambda)}=\frac{Z^+_{G\backslash\{u\},v}}{Z_{G\backslash\{u\}}}=P_{G\backslash\{u\}, v}(\lambda)$$
and
$$P^{u^+}_{G, v}(\lambda)=\frac{Z^{+,+}_{G,u,v}(\lambda)}{Z^+_{G,u}(\lambda)}=\frac{\lambda Z^+_{G\backslash N[u],v}}{\lambda Z_{G\backslash N[u]}}=P_{G\backslash N[u], v}(\lambda).$$
 
In fact, it is easy to see that for any graph $G$ and any feasible partial configuration $\sigma_\Lambda$ on some $\Lambda \subseteq U$, there exists a subgraph $G'$ of $G$ such that $P^{\sigma_\Lambda}_{G,v}(\lambda)=P_{G', v}(\lambda)$. Since the family $\mathcal{G}$ is closed under taking subgraphs, $G' \in \mathcal{G}$. Thus, in order to establish point-to-point LDC, it suffices to prove it for a tree $G\in \mathcal{G}$ where no vertices are pinned. 
%By Theorem~\ref{prop:SAW}, we may consider the  SAW tree of $G$ rooted at the vertex $v$, which is included in $\mathcal{G}$.
%Still by the SAW-tree construction, we may assume $G$ is a tree. 
We observe that
\begin{equation*}
\begin{aligned}
     & Z_G Z_{G\setminus\{u,v\}}-Z_{G\setminus\{u\}}Z_{G\setminus\{v\}} \\ = &  (Z^{-,-}_{G,u,v}+Z^{+,+}_{G,u,v}+Z^{-,+}_{G,u,v}+Z^{+,-}_{G,u,v})Z^{-,-}_{G,u,v} 
    - (Z^{-,-}_{G,u,v}+Z^{-,+}_{G,u,v})(Z^{-,-}_{G,u,v}+Z^{+,-}_{G,u,v}) \\
  = &   Z_{G,u,v}^{+,+}Z_{G,u,v}^{-,-}-Z_{G,u,v}^{+,-}Z_{G,u,v}^{-,+}.
 \end{aligned}
\end{equation*}
Here we omit the variable $\lambda$ for simplicity. Then, we have
\begin{equation}
    \begin{aligned}
        P_{G,v}(\lambda)-P^{u^-}_{G, v}(\lambda)
        &=\dfrac{Z^+_{G,v}(\lambda)}{Z_{G}(\lambda)}-\dfrac{Z^{+,-}_{G,v,u}(\lambda)}{Z^-_{G,u}(\lambda)}\\
        &=\dfrac{Z_{G,u,v}^{+,+}(\lambda)Z_{G,u,v}^{-,-}(\lambda)-Z_{G,u,v}^{+,-}(\lambda)Z_{G,u,v}^{-,+}(\lambda)}{Z_G(\lambda) Z^-_{G,u}(\lambda)} \\
        &=\frac{-(-\lambda)^{d(u,v)+1}Z_{G\setminus p_{uv}}(\lambda)Z_{G\setminus N[p_{uv}]}(\lambda)}{Z_G(\lambda) Z_{G\setminus\{u\}}(\lambda)}\\
        & =\lambda^{d(u,v)}H_1(\lambda)
    \end{aligned}
\end{equation}
    where $H_1(\lambda)$ is a rational function on $U$ which is analytic near $0$. Thus, the coefficients of the term $\lambda^k$ for $0\leq k \leq d_G(u,v)-1$ in the Taylor series  of them are the same. Also, we have
\begin{equation}
    \begin{aligned}
        P_{G,v}(\lambda)-P^{u^+}_{G, v}(\lambda)
        &=\dfrac{Z^+_{G,v}(\lambda)}{Z_{G}(\lambda)}-\dfrac{Z^{+,+}_{G,v,u}(\lambda)}{Z^+_{G,u}(\lambda)}\\
        &=\dfrac{-(Z_{G,u,v}^{+,+}(\lambda)Z_{G,u,v}^{-,-}(\lambda)-Z_{G,u,v}^{+,-}(\lambda)Z_{G,u,v}^{-,+}(\lambda))}{Z_G(\lambda) Z^+_{G,u}(\lambda)} \\
        &=\frac{(-\lambda)^{d(u,v)+1}Z_{G\setminus p_{uv}}(\lambda)Z_{G\setminus N[p_{uv}]}(\lambda)}{\lambda Z_G(\lambda) Z_{G\setminus N[u]}(\lambda)}\\
        &=\lambda^{d(u,v)}H_2(\lambda)
    \end{aligned}
\end{equation}
    where $H_2(\lambda)$ is a rational function on $U$ which is analytic near $0$. Thus, the coefficients of the term $\lambda^k$ for $0\leq k \leq d_G(u,v)-1$ in the Taylor series of them are the same.
\end{proof}

\begin{remark}
    The above proof of LDC for the hard-core model using the  Christoffel-Darboux type identity may somewhat answer the following interesting question: \emph{if the partition function of the hard-core model is zero-free on a complex neighborhood of some point $\lambda_0\neq 0$, can we get SSM for this zero-free region?} 
    The answer would be yes if the LDC property could hold for rational functions $P_G^v(\lambda)$ expanded as Taylor series near $\lambda=\lambda_0$ instead of $\lambda=0$. 
    However, the Christoffel-Darboux type identity implies that such a LDC property does not hold. Otherwise, for any graph $G$ and any vertex $u$, we would have $P_{G,v}(\lambda)-P^{u^+}_{G, v}(\lambda)=(\lambda-\lambda_0)^{d(u,v)}H(\lambda)$ for some  function $H(\lambda)$ analytic near $\lambda_0$.
    One can easily check that this is false. 
   Thus,  it is highly possible that a zero-free region not containing $0$ cannot imply SSM (at least not by our approach).   
\end{remark}

Now, we extend the above Christoffel-Darboux type identity to  general 2-spin systems on trees. Actually, we can extend it to the more general $q$-spins system on trees (see Appendix~\ref{sec:CDqspin}). 

\begin{theorem}[Theorem~\ref{CDgeneralintro} restated]\label{CDgeneral}
    Suppose that $T$ is a tree, $\sigma_{\Lambda}$ is a partial configuration on some $\Lambda \subseteq V$ which may be empty,  and $u$ and $v$ are two distinct vertices in $V\backslash \Lambda$.
    Let $p_{uv}$ denote  the unique path in $T$ connecting $u$ and $v$, $V(p_{uv})$ denote the set of points in $p_{uv}$, $N[p_{uv}]$ denote the set of neighbors of points in $p_{uv}$,  $N[p_{uv}]\setminus V(p_{uv})=\{v_1,\cdots,v_n\}$, and $T_i$ denote the component of $v_i$ in $T\setminus p_{uv}$. 
    Then for the partition function of the 2-spin system defined on $T$ conditioning on $\sigma_{\Lambda}$, we have {(we omit the argument $(\beta, \gamma, \lambda)$ in the second line for simplicity)}
    \[\begin{aligned}
& Z_{T,u,v}^{\sigma_{\Lambda},+,+}(\beta,\gamma,\lambda)Z_{T,u,v}^{\sigma_{\Lambda}, -,-}(\beta,\gamma,\lambda) 
-Z_{T,u,v}^{\sigma_{\Lambda}, +,-}(\beta,\gamma,\lambda)Z_{T,u,v}^{\sigma_{\Lambda}, -,+}(\beta,\gamma,\lambda) \\
=&\begin{cases}(\beta\gamma-1)^{d(u,v)}\lambda^{d(u,v)+1}\prod\limits_{i=1}^{n}(\beta Z_{T_i,v_i}^{\sigma_{\Lambda},+}+Z_{T_i,v_i}^{\sigma_{\Lambda},-})(Z_{T_i,v_i}^{\sigma_{\Lambda},+} +\gamma Z_{T_i,v_i}^{\sigma_{\Lambda},-}) & , \text{ if } V(p_{uv})\cap\Lambda=\emptyset \\
0 & , \text{ if } V(p_{uv})\cap \Lambda\neq\emptyset
\end{cases}.\end{aligned}\]
\end{theorem}

\begin{proof}
In the following proof and Remark~\ref{rmk:cd}, we omit the argument $(\beta, \gamma, \lambda)$ of the partition functions for simplicity. 
We prove the equality by induction on $k=d(u,v)$ when $V(p_{uv})\cap\Lambda=\emptyset$. When $k=1,$ $u$ and $v$ are adjacent, Denote $u,v_1,\cdots,v_l$ neighbors of $v$ and $v,v_{l+1},\cdots,v_{n}$ neighbors of $u$, as illustrated in the following figure.

\begin{figure}[!hbtp]\label{fig:ldc}
\centering
	\includegraphics[scale=0.18]{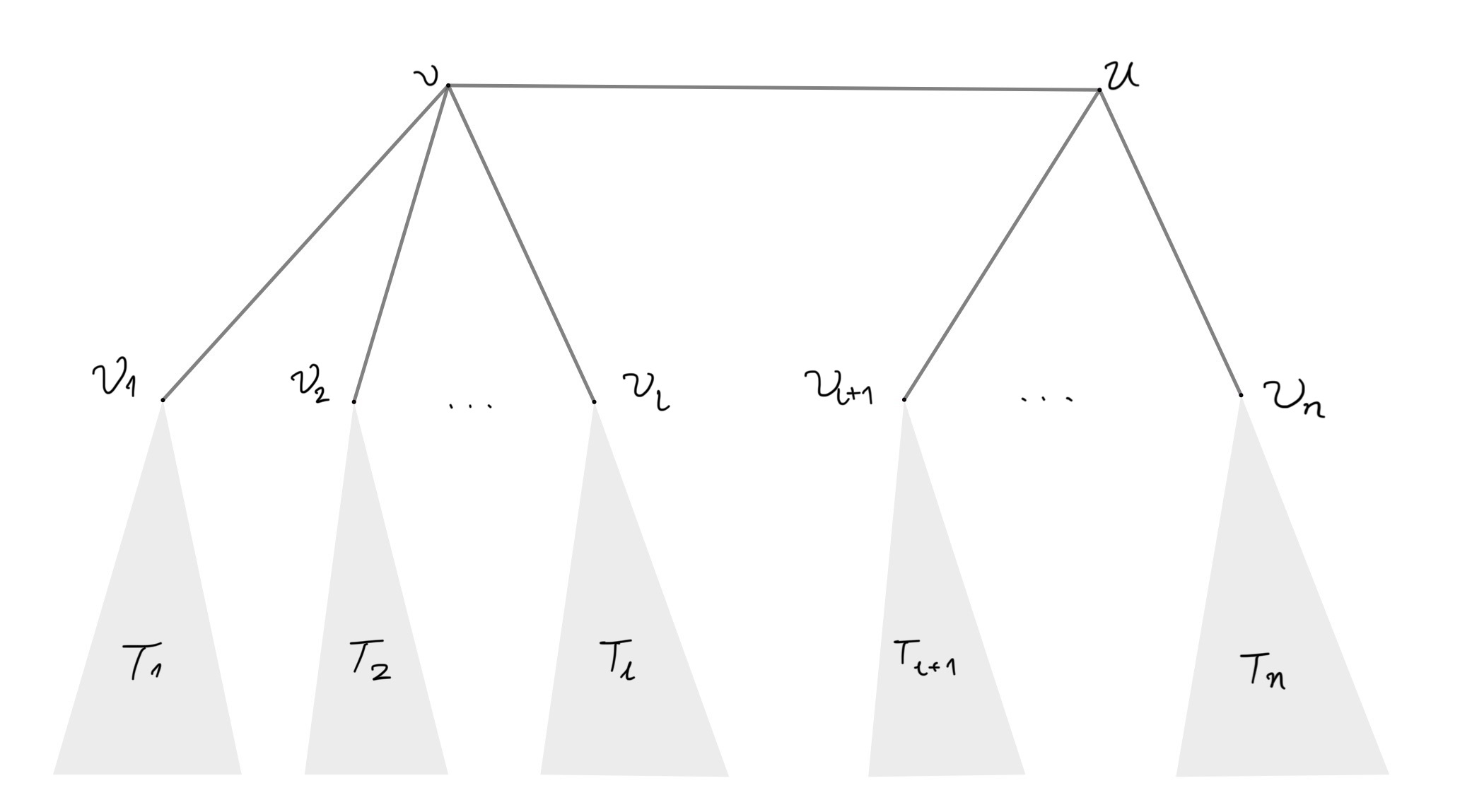}
	\label{fig:gutmangen}
\end{figure}
Then, we have
    \[Z_{T,u,v}^{\sigma_{\Lambda},+,+}=\lambda^2\beta\prod\limits_{i=1}^{n}(\beta Z_{T_i,v_i}^{\sigma_{\Lambda},+} +Z_{T_i,v_i}^{\sigma_{\Lambda},-});\]
    \[Z_{T,u,v}^{\sigma_{\Lambda},-,-}=\gamma\prod\limits_{i=1}^{n}(Z_{T_i,v_i}^{\sigma_{\Lambda},+} +\gamma Z_{T_i,v_i}^{\sigma_{\Lambda},-}); \]
    \[Z_{T,u,v}^{\sigma_{\Lambda},-,+}=\lambda\prod\limits_{i=1}^{l}(\beta Z_{T_i,v_i}^{\sigma_{\Lambda},+} +Z_{T_i,v_i}^{\sigma_{\Lambda},-})\prod\limits_{i=l+1}^{n}(Z_{T_i,v_i}^{\sigma_{\Lambda},+} +\gamma Z_{T_i,v_i}^{\sigma_{\Lambda},-});\]
    \[Z_{T,u,v}^{\sigma_{\Lambda},+,-}=\lambda\prod\limits_{i=1}^{l}(Z_{T_i,v_i}^{\sigma_{\Lambda},+} +\gamma Z_{T_i,v_i}^{\sigma_{\Lambda},-})\prod\limits_{i=l+1}^{n}(\beta Z_{T_i,v_i}^{\sigma_{\Lambda},+} +Z_{T_i,v_i}^{\sigma_{\Lambda},-}).\]
Thus, we deduce that
    \[
Z_{T,u,v}^{\sigma_{\Lambda},+,+}Z_{T,u,v}^{\sigma_{\Lambda},-,-}-Z_{T,u,v}^{\sigma_{\Lambda},+,-}Z_{T,u,v}^{\sigma_{\Lambda},-,+}
 =(\beta\gamma-1)\lambda^{2}\prod\limits_{i=1}^{n}(\beta Z_{T_i,v_i}^{\sigma_{\Lambda},+} +Z_{T_i,v_i}^{\sigma_{\Lambda},-})(Z_{T_i,v_i}^{\sigma_{\Lambda},+} +\gamma Z_{T_i,v_i}^{\sigma_{\Lambda},-}).
\]

Suppose now that $k\ge 2$ and the assertion is true for $k-1$. Suppose that $v_0,v_1,\cdots,v_l$ are neighbors of $v$ with $v_0$ on $p_{uv}.$ Let $T_0$ denote the component of $T\setminus\{v\}$ containing $v_0$ and $u$. Then 
\[Z_{T,u,v}^{\sigma_{\Lambda},+,+}=\lambda(\beta Z_{T_0,v_0,u}^{\sigma_{\Lambda},+,+}+Z_{T_0,v_0,u}^{\sigma_{\Lambda},-,+})\prod\limits_{i=1}^{l}(\beta Z_{T_i,v_i}^{\sigma_{\Lambda},+} +Z_{T_i,v_i}^{\sigma_{\Lambda},-});\]
\[ Z_{T,u,v}^{\sigma_{\Lambda},-,-}=(Z_{T_0,v_0,u}^{\sigma_{\Lambda},+,-}+\gamma Z_{T_0,v_0,u}^{\sigma_{\Lambda},-,-})\prod\limits_{i=1}^{l}(Z_{T_i,v_i}^{\sigma_{\Lambda},+} +\gamma Z_{T_i,v_i}^{\sigma_{\Lambda},-}); \]
\[Z_{T,u,v}^{\sigma_{\Lambda},-,+}=\lambda(\beta Z_{T_0,v_0,u}^{\sigma_{\Lambda},+,-}+Z_{T_0,v_0,u}^{\sigma_{\Lambda},-,-})\prod\limits_{i=1}^{l}(\beta Z_{T_i,v_i}^{\sigma_{\Lambda},+} +Z_{T_i,v_i}^{\sigma_{\Lambda},-});\]
\[Z_{T,u,v}^{\sigma_{\Lambda},+,-}=( Z_{T_0,v_0,u}^{\sigma_{\Lambda},+,+}+\gamma Z_{T_0,v_0,u}^{\sigma_{\Lambda},-,+})\prod\limits_{i=1}^{l}(Z_{T_i,v_i}^{\sigma_{\Lambda},+} +\gamma Z_{T_i,v_i}^{\sigma_{\Lambda},-}),\]
and hence,
 \[\begin{aligned}
& Z_{T,u,v}^{\sigma_{\Lambda},+,+}Z_{T,u,v}^{\sigma_{\Lambda},-,-} -Z_{T,u,v}^{\sigma_{\Lambda},+,-}Z_{T,u,v}^{\sigma_{\Lambda},-,+} \\ = &(\beta\gamma-1)\lambda (Z_{T_0,u,v_0}^{,\sigma_{\Lambda},+,+}Z_{T_0,u,v_0}^{\sigma_{\Lambda},-,-}  -Z_{T_0,u,v_0}^{\sigma_{\Lambda},+,-}Z_{T_0,u,v_0}^{\sigma_{\Lambda},-,+})\prod\limits_{i=1}^{l}(\beta Z_{T_i,v_i}^{\sigma_{\Lambda},+} +Z_{T_i,v_i}^{\sigma_{\Lambda},-})(Z_{T_i,v_i}^{\sigma_{\Lambda},+} +\gamma Z_{T_i,v_i}^{\sigma_{\Lambda},-}).
\end{aligned}\]
 Notice that $d(u,v_0)=k-1$.
 Then, we can apply the induction hypothesis, which will finish the induction proof. 

When $V(p_{uv})\cap\Lambda\neq\emptyset$, we argue by induction on $l=d(u,V(p_{uv})\cap\Lambda)+d(v,V(p_{uv})\cap\Lambda)\ge 2.$ The the basic case $l=2$ is easy to check and the proof of the inductive step is same as the above. 
\end{proof}

\begin{remark}\label{rmk:cd}
    Since $$Z^{\sigma_{\Lambda}}_{T}=Z^{\sigma_{\Lambda}, +, +}_{T, u, v}+Z^{\sigma_{\Lambda}, +, -}_{T, u, v}+Z^{\sigma_{\Lambda}, -, +}_{T, u, v}+Z^{\sigma_{\Lambda}, -, -}_{T, u, v},$$ $$Z^{\sigma_{\Lambda}, +}_{T, u}=Z^{\sigma_{\Lambda}, +, +}_{T, u, v}+Z^{\sigma_{\Lambda}, +, -}_{T, u, v} \text{  and  } Z^{\sigma_{\Lambda}, -}_{T, u}=Z^{\sigma_{\Lambda}, -, -}_{T, u, v}+Z^{\sigma_{\Lambda}, -, +}_{T, u, v},$$ 
we have 
    \[\begin{aligned}       Z^{\sigma_{\Lambda}}_{T}Z^{\sigma_{\Lambda}, +, +}_{T, u, v}-Z^{\sigma_{\Lambda}, +}_{T, u}Z^{\sigma_{\Lambda}, +}_{T,v} & =Z^{\sigma_{\Lambda}}_{T}Z^{\sigma_{\Lambda}, -, -}_{T, u, v}-Z^{\sigma_{\Lambda}, -}_{T, u}Z^{\sigma_{\Lambda}, -}_{T,v}\\ 
& =Z_{T,u,v}^{\sigma_{\Lambda},+,+}Z_{T,u,v}^{\sigma_{\Lambda}, -,-} 
-Z_{T,u,v}^{\sigma_{\Lambda}, +,-}Z_{T,u,v}^{\sigma_{\Lambda}, -,+}. 
    \end{aligned}\]
\end{remark}

\begin{lemma}\label{zftoLDC-tree}
Fix complex parameters $\beta$ and $\gamma$, a complex neighborhood $U$ of $0$, and a family of trees $\mathcal{G}$. If the partition function $Z^{\sigma_\Lambda}_{G, \beta, \gamma}(\lambda)$
 is zero-free on $U\backslash\{0\}$, then the corresponding 2-spin system satisfies point-to-point LDC on $U$. 
\end{lemma}
\begin{proof}
For $G=(V,E)\in\mathcal{G}$, feasible configuration $\sigma_{\Lambda}\in\{0,1\}^{\Lambda}$ on some $\Lambda\subseteq V$, and two vertices $u,v$ proper to $\sigma_\Lambda$, consider the rational functions $P_{G,v}^{\sigma_\Lambda}(\lambda)$, $P_{G,v}^{\sigma_\Lambda, u^+}(\lambda)$ and $P_{G,v}^{\sigma_\Lambda, u^-}(\lambda)$. 
One can easily check that they are analytic on $U$.

%By the SAW tree construction, it suffices to consider the case that 
Since $G$ is a tree, 
by Theorem~\ref{CDgeneral} and Remark~\ref{rmk:cd}, 
\[\begin{aligned}
    P_{G,v}^{\sigma_\Lambda}(\lambda)-P_{G,v}^{\sigma_\Lambda, u^+}(\lambda) & =\frac{Z_{G,u}^{\sigma_\Lambda,+}(\lambda)Z_{G,v}^{\sigma_\Lambda,+}(\lambda)-Z_{G,v,u}^{\sigma_\Lambda,+,+}(\lambda)Z_{G}^{\sigma_\Lambda}(\lambda)}{Z_{G}^{\sigma_\Lambda}(\lambda)Z_{G,u}^{\sigma_\Lambda,+}(\lambda)} \\
    & =\lambda^{d_G(u,v)+1}H(\lambda)
\end{aligned}\]
for some $H(\lambda)$ analytic near $0$. Thus, the coefficients of the term $\lambda^k$ for $0\leq k \leq d_G(u,v)$ in the Taylor series  of $P_{G,v}^{\sigma_\Lambda}(\lambda)$ and $P_{G,v}^{\sigma_\Lambda, u^+}(\lambda)$ are equal.
Similar result holds for $P_{G,v}^{\sigma_\Lambda}(\lambda)$ and $P_{G,v}^{\sigma_\Lambda, u^-}(\lambda)$.
\end{proof}

Combining Lemmas~\ref{PLDCtoLDC}, \ref{lem:tree-enough} and~\ref{zftoLDC-tree}, we have the following result. 

\begin{lemma}\label{zftoLDC}
Fix complex parameters $\beta$ and $\gamma$, a complex neighborhood $U$ of $0$, and a family of graphs $\mathcal{G}$ closed under {\rm SAW} tree constructions. If the partition function $Z^{\sigma_\Lambda}_{G, \beta, \gamma}(\lambda)$
 is zero-free on $U\backslash\{0\}$, then the corresponding 2-spin system satisfies LDC on $U$. 
\end{lemma}

Notice that in the Christoffel-Darboux type identity for the 2-spin system, besides the ``factor'' $\lambda^{d(u,v)+1}$, one can also take a ``factor'' $(\beta\gamma-1)^{d(u,v)}$ out of $Z_{T,u,v}^{\sigma_{\Lambda},+,+}Z_{T,u,v}^{\sigma_{\Lambda}, -,-} 
-Z_{T,u,v}^{\sigma_{\Lambda}, +,-}Z_{T,u,v}^{\sigma_{\Lambda}, -,+}$.
Thus, we can prove LDC for $\beta$ near $1/\gamma$ similarly. 
The proof of LDC for $\gamma$ near $1/\beta$ is symmetric.

%In Section~\ref{betagamma}, we will use this fact to prove that a zero-free region of $\beta$ or $\gamma$ containing the point $\beta=1/\gamma$ or $\gamma=1/\beta$ respectively implies SSM. 

\begin{lemma}\label{lem:ldcbeta}
    Fix complex parameters $\gamma\neq 0$ and $\lambda$. Let $\mathcal{G}$ be a family of graphs closed under {\rm SAW} tree constructions and $U$ be a complex neighborhood of $1/\gamma$ such that the partition function $Z^{\sigma_{\Lambda}}_{G,\gamma, \lambda}(\beta)$ is zero-free on $U$.
    Then the corresponding 2-spin system satisfies LDC for $\beta$ near $1/\gamma$.
\end{lemma}
\begin{proof}
It suffices to prove that the system satisfies point-to-point LDC for $\beta$ near $1/\gamma$.
    By the SAW tree construction and  an analogue of Lemma~\ref{lem:tree-enough} for $\beta$ near $1/\gamma$, we only need to consider the case that $G$ is a tree.
    Then by Theorem~\ref{CDgeneral}, we have
\[\begin{aligned}
    P_{G,v}^{\sigma_\Lambda}(\beta)-P_{G,v}^{\sigma_\Lambda, u^+}(\beta)
    & =\frac{Z_{G,u}^{\sigma_\Lambda,+}(\beta)Z_{G,v}^{\sigma_\Lambda,+}(\beta)-Z_{G,v,u}^{\sigma_\Lambda,+,+}(\beta)Z_{G}^{\sigma_\Lambda}(\beta)}{Z_{G}^{\sigma_\Lambda}(\beta)Z_{G,u}^{\sigma_\Lambda,+}(\beta)}\\
    & =(\beta-1/\gamma)^{d_G(u,v)}\gamma^{d_G(u,v)}H(\beta)
\end{aligned}\]
for some function $H(\beta)$ analytic near $1/\gamma$.
The same form of equation holds for $P_{G,v}^{\sigma_\Lambda}(\beta)-P_{G,v}^{\sigma_\Lambda, u^-}(\beta)$.
We are done with the proof of LDC for $\beta$ near $1/\gamma$. 
\end{proof}

In addition, for the Ising model, $\beta=\gamma$, $\beta\gamma=1$ implies $\beta=\gamma=1$ or $-1$. 
Thus, similarly, we can define and prove LDC for $\beta$ near $1$ or $-1$ for the Ising model.  

\begin{definition}[LDC for $\beta$ near $1$ or $-1$]

Fix a complex parameter $\lambda$, a neighborhood $U$ of $1$ (or respectively $-1$), and a family of graphs $\mathcal{G}$.
The corresponding Ising model (defined on $\mathcal{G}$ specified with the fixed parameter $\lambda$, and a variable $\beta \in U$) is said to satisfy LDC for $\beta$ near $1$ (or respectively $-1$) if for
any $G=(V,E)\in\mathcal{G}$,
any feasible configurations $\sigma_{\Lambda_1}\in\{0,1\}^{\Lambda_1}$ and $\tau_{\Lambda_2}\in\{0,1\}^{\Lambda_2}$, and any vertex $v$ proper to $\sigma_{\Lambda_1}$ and $\tau_{\Lambda_2}$, the functions $P_{G,v}^{\sigma_{\Lambda_1}}(\beta)$ and $P_{G,v}^{\tau_{\Lambda_2}}(\beta)$ are analytic on $U$, and their corresponding Taylor series $\sum^\infty_{i=0}a_i(\beta-1)^i$ and $\sum^\infty_{i=0}b_i(\beta-1)^i$  near $\beta=1$ (or respectively $\sum^\infty_{i=0}a_i(\beta+1)^i$ and $\sum^\infty_{i=0}b_i(\beta+1)^i$ near  $\beta=-1$)  satisfy $a_i=b_i$ for $0\leq i \leq d_G(v,\sigma_{\Lambda_1}\neq \tau_{\Lambda_2})-1$.
\end{definition}

\begin{lemma}\label{lem:ldcising}
    Fix  $\lambda \in \mathbb{C}$. Let $\mathcal{G}$ be a family of graphs closed under {\rm SAW} tree constructions and $U$ be a complex neighborhood of $1$ (or respectively $-1$) such that $Z^{\sigma_{\Lambda}}_{G,\gamma, \lambda}(\beta)$ is zero-free on $U$.
    Then the corresponding Ising model satisfies LDC for $\beta$ near $1$ (or respectively $-1$).
\end{lemma}

\section{Zero-freeness implies SSM}\label{sec:together}

Now, we are ready to show that zero-freeness implies SSM. 
We first apply Montel's theorem (Theorem \ref{thm:montel}) to get a uniform bound for the multivariate function $P^{\sigma_\Lambda}_{G,v}(\beta, \gamma, \lambda)$, which extends Lemma 5 of \cite{Guus2021zerofreetossm}.

\begin{lemma}\label{bdcircle}
Let $\mathcal{G}$ be
a family of graphs,  and $U_1, U_2, U_3\subseteq \mathbb{C}$ be quasi-neighborhoods of some $\beta_0, \gamma_0, \lambda_0\in \mathbb{C}$ where $U_1\times U_2\times U_3\neq \{(\beta_0, \gamma_0, \lambda_0)\}$ (i.e., at least one of  $U_1, U_2$, and $U_3$ is a complex region). 
Denote the set $((U_1\times U_2) \backslash\{(0, 0)\})\times (U_3\backslash\{0\})$ by  $\mathbf U$. 
Suppose that the partition function of the 2-spin system defined on $\mathcal{G}$ is zero-free on $\mathbf{U}$, and one of the following conditions holds.
\begin{enumerate}
    \item $\beta_0, \gamma_0, \lambda_0$ are non-negative real numbers that are not all zero;
    \item $\mathcal{G}$ is
a family of graphs of bounded degree and $\beta_0\gamma_0=1$.
\end{enumerate}
Then for any compact set $S\subseteq  \mathbf U$, there exists an $M>0$ such that for every $G\in \mathcal{G}$, any feasible partial configuration $\sigma_\Lambda$, any vertex $v$ proper to $\sigma_\Lambda$,  we have $|P^{\sigma_\Lambda}_{G,v}(\beta, \gamma, \lambda)|\leq M$ for all $(\beta, \gamma, \lambda)\in S$.
\end{lemma}

\begin{proof}
    We give the proof when $U_1,U_2$ are neighborhoods of some points.
    For the case that some of them are singletons, the proof is identical using Montel's theorem (Theorem \ref{thm:montel}) with fewer variables.
    
    Suppose to the contrary that there exists a sequence of tuples $(G_n\in\mathcal{G},v_n,\Lambda_n,\sigma^{(n)}_{\Lambda_n},(\beta_n,\gamma_n,\lambda_n)\in S)_{n\in\N}$ such that\[\left|P_{G_n,v_n}^{\sigma^{(n)}_{\Lambda_n}}(\beta_n,\gamma_n,\lambda_n)\right|\ge n.\]
    We denote by $\{f_n(\beta, \gamma, \lambda)\}_{n\in\N}$ the sequence of functions  $\{P_{G_n,v_n}^{\sigma^{(n)}_{\Lambda_n}}(\beta,\gamma,\lambda)\}_{n\in\N}$.
    
  Since the partition function of the 2-spin system defined on $\mathcal{G}$ is zero-free on $\mathbf{U}$, for any $G=(V,E)\in\mathcal{G}$, any feasible partial configuration $\sigma_{\Lambda}$ on some $\Lambda\subseteq V$, and any vertex $v$ proper to $\sigma_\Lambda$,
    the function $P^{\sigma_\Lambda}_{G,v}(\beta, \gamma, \lambda)$ is holomorphic on $\bf U$ (which is connected since $U_3\setminus\{0\}$ and $U_1\times U_2\setminus\{(0,0)\}$ are connected) and avoids the values $0$ and $1$. Thus all such functions form a normal family by Montel's theorem (Theorem \ref{thm:montel}). 

Now we consider the two conditions. If  $\beta_0,\gamma_0,\lambda_0\geq 0$ and they are not all zero, then for any rational function $P^{\sigma_\Lambda}_{G,v}$, its value given the input $(\beta_0,\gamma_0,\lambda_0)$  is a probability which lies in the interval $(0,1)$.
Otherwise, $\mathcal{G}$ is a family of graphs the degrees of which are bounded by some $\Delta$ and $\beta_0\gamma_0=1$.
Then, one can check that for any rational function $P^{\sigma_\Lambda}_{G,v}$, its value given the input $(\beta_0,\gamma_0,\lambda_0)$ is $\frac{\lambda_0\beta_0^{d_{v}}}{\lambda_0\beta_0^{d_{v}}+1}$ where $d_v$ is the degree of $v$ in $G$. 
Note that $P^{\sigma_\Lambda}_{G,v}(\beta_0,\gamma_0,\lambda_0)\neq \infty$ since the partition function is nonzero on $\mathbb{U}$. 
Since $d_v$ is bounded by the maximum degree $\Delta$, 
 for all rational functions $P^{\sigma_\Lambda}_{G,v}$, 
 we have $|P^{\sigma_\Lambda}_{G,v}(\beta_0, \gamma_0, \lambda_0)|\leq \max_{0\leq d_v\leq  \Delta}\left|\frac{\lambda_0\beta_0^{d_{v}}}{\lambda_0\beta_0^{d_{v}}+1}\right|<\infty$. 
    Thus, in both cases, any sequence in the family of  rational functions $P^{\sigma_\Lambda}_{G,v}$ can not converge locally uniformly to $\infty$. Then by Definition~\ref{def:normal}, every sequence has a sub-sequence converging locally uniformly to a holomorphic function. 
%     Consider the sequence of functions $\{P_{G_n,v_n}^{\sigma^{(n)}_{\Lambda_n}}(\beta,\gamma,\lambda)\}_{n\in\N}$ denoted by $\{f_n\}_{n\in\N}$.

     Thus, the above sequence $\{f_n\}_{n\in \mathbb{N}}$ has a sub-sequence $\{f_{n_i}\}_{i\in \mathbb{N}}$ converging  locally uniformly to some holomorphic function $f$  on $\mathbf{U}$. 
     Since $S$ is a compact subset of $\mathbf{U}$, the convergence is uniform on $S$, and the function $f$ is bounded by some $C>0$ on $S$. 
     Then, there exists some $n_k\in \mathbb{N}$ such that for all $n_i\geq n_k$ , we have $|f_{n_i}(\beta, \gamma, \lambda)|<2C$ on $S$.  
     Clearly, this contradicts that for all $n\in \mathbb{N}$,  $|f_n(\beta_n, \gamma_n, \lambda_n)|\geq n$ where $(\beta_n, \gamma_n, \lambda_n)\in S$. 
\end{proof}

\begin{remark}\label{remark-bound-circle}
    In particular, when the set $U_3$ in Lemma~\ref{bdcircle} is a complex neighborhood of $0$, we may take the compact set $S$ to be $\{\beta_1\}\times\{\gamma_1\}\times \partial\mathbb{D}_\rho$ for some arbitrary $\beta_1\in U_1$, $\gamma_1\in U_2$, and $\partial\mathbb{D}_\rho\subseteq U_3$. Then, we get a uniform bound for $P^{\sigma_\Lambda}_{G,v, \beta_1, \gamma_1}(\lambda)=P^{\sigma_\Lambda}_{G,v}(\beta_1, \gamma_1, \lambda)$ on the circle $\partial\mathbb{D}_\rho$.
    Similarly, 
    by taking $S$ to be $\partial\mathbb{D}_{\rho}(1/\gamma_1)\times \{\gamma_1\} \times \{\lambda_1\}$,
    one can get a  uniform bound for $P^{\sigma_\Lambda}_{G,v,_{\gamma_1,\lambda_1}}(\beta)$ on the circle $\partial\mathbb{D}_{\rho}(1/\gamma_1)$, and by taking $S$ to be $\{\beta_1\} \times  \partial\mathbb{D}_{\rho}(1/\beta_1)\times \{\lambda_1\}$,
    one can get a  uniform bound for $P^{\sigma_\Lambda}_{G,v,_{\beta_1,\lambda_1}}(\gamma)$ on the circle $\partial\mathbb{D}_{\rho}(1/\beta_1)$.
\end{remark}

Then, combining LDC and the above uniform bound, we can prove SSM on a disk region.
Still, we give the proof of SSM for $\lambda$ as an illustration. 
The proof of SSM for $\beta$ or $\gamma$ is similar.

%By Remark~\ref{remark-bound-circle} and Lemma~\ref{lem:bound strip}, we have the following result. 

\begin{lemma}\label{zftossmdisk}
    Let $\mathcal{G}$ be
a family of graphs closed under {\rm SAW} tree constructions, $U_1$ and $U_2$ be quasi-neighborhoods of some $\beta_0, \gamma_0 \in \mathbb{C}$ respectively, and
$U_3$ be a neighborhood of $0$. 
Denote $((U_1\times U_2) \backslash\{(0, 0)\})\times (U_3\backslash\{0\})$ by  $\mathbf U$. 
Suppose that the partition function of the 2-spin systems defined on $\mathcal{G}$ is zero-free on $\mathbf{U}$, and one of the two conditions of Lemma~\ref{bdcircle} holds.
Then, for any $(\beta, \gamma, \lambda) \in \mathbf U$ satisfying $\lambda \in \mathbb{D}_\rho \subseteq U_3$ for some disk $\mathbb{D}_\rho$ centered at $0$, the corresponding 2-spin system exhibits SSM  with rate $r=\rho/|\lambda|$. 

    In particular,  when $U_3$ itself is a disk $\mathbb{D}_\rho$,  every $(\beta, \gamma, \lambda) \in \mathbf U$ satisfies $\lambda \in \mathbb{D}_\rho = U_3$.
    Thus, for any $(\beta, \gamma, \lambda) \in \mathbf U$, the  2-spin system exhibits SSM.
\end{lemma}

\begin{proof}
For any fixed $\beta\in U_1$ and $\gamma \in U_2$, the 2-spin system defined on $\mathcal{G}$ with fixed $\beta$, $\gamma$ and a variable $\lambda\in U_3$ is zero-free on $U_3\backslash\{0\}$.
Hence, it satisfies LDC by Lemma \ref{zftoLDC},  and  for any circle $\partial \mathbb{D}_\rho\subseteq U_3$ it has a uniform bound on $\partial \mathbb{D}_\rho$  by Remark \ref{remark-bound-circle}. 
 For some arbitrary $G=(V,E)\in\mathcal{G}$,
 {some} feasible partial configurations $\sigma_{\Lambda_1}$ and $\tau_{\Lambda_2}$, and some vertex $v$ proper to $\sigma_{\Lambda_1}$ and $\tau_{\Lambda_2}$, consider the rational functions $P(\lambda)=P_{G,v}^{\sigma_{\Lambda_1}}(\lambda)$ and $Q(\lambda)=P_{G,v}^{\tau_{\Lambda_2}}(\lambda)$. By LDC, these two functions are analytic and the coefficients of the first $d_G(v,\{\sigma_{\Lambda_1}\neq\tau_{\Lambda_2}\})$ terms in their Taylor series are equal. 
    Also, by a uniform bound on  the circle $\partial{\mathbb{D}_\rho}$, there exists $M>0$ such that 
 $|P(\lambda)|\le M$ and $|Q(\lambda)|\le M$ for all $\lambda\in \partial \mathbb{D}_{\rho}$, and then the result follows directly from Lemma~\ref{lem:bound strip}.
%Thus,  by Lemma \ref{SSMfromLDCandbdcircle},  for any $(\beta, \gamma, \lambda) \in \mathbf U$ satisfying $\lambda \in \mathbb{D}_\rho \subseteq U_3$ for some $\mathbb{D}_\rho$, the 2-spin system exhibits SSM. 
\end{proof}

%\subsection{From an arbitrary region to the unit disk}

By Lemma~\ref{zftossmdisk}, when the zero-free region of $\lambda$ is a disk, the 2-spin system exhibits SSM for any $\lambda$ in the disk. 
 Using Riemann mapping theorem,  we are able to map any region $U$ containing $0$ to the unit disk $\mathbb{D}_1$ with $0$ as a fixed point. 
 Finally, we can drop the restriction of $\lambda \in \mathbb{D}_\rho \subseteq U_3$ for some disk $\mathbb{D}_\rho$ as in the statement of  Lemma~\ref{zftossmdisk}, and get the following result: zero-freeness implies SSM for any complex neighborhood of $\lambda=0$.

 \begin{theorem}\label{mainany} 
Let $\mathcal{G}$ be
a family of graphs closed under {\rm SAW} tree constructions, 
$U_1$ and $U_2$ be quasi-neighborhoods of some $\beta_0, \gamma_0 \in \mathbb{C}$ respectively, and 
$U_3$ be a neighborhood of $0$.
Denote $((U_1\times U_2) \backslash\{(0, 0)\})\times (U_3\backslash\{0\})$ by $\mathbf U$. 
Suppose that the partition function of the 2-spin systems defined on $\mathcal{G}$ is zero-free on $\mathbf{U}$, and  one of the two conditions of Lemma~\ref{bdcircle} holds.
Then, for any $(\beta, \gamma, \lambda) \in \bf U$, the corresponding 2-spin system exhibits SSM. 
 \end{theorem}
\begin{proof}
    By Riemann mapping theorem (Theorem \ref{riemann}), there exists a conformal map $f:\mathbb{D}_1\to U_3$ with $f(0)=0$. 
    
    For a rational function $P^{\sigma_\Lambda}_{G,v}(\beta, \gamma, \lambda)$ where $v$ is proper to the feasible partial configuration $\sigma_\Lambda$,
    we consider the composition $g_{G,v}^{\sigma_{\Lambda}}(\beta, \gamma, \lambda)=P_{G,v}^{\sigma_\Lambda}(\beta, \gamma, f(\lambda))$.
    Clearly,  $g_{G,v}^{\sigma_{\Lambda}}$ is holomorphic  on ${\bf U'}=(U_1\times U_2 \backslash\{(0, 0)\})\times(\mathbb{D}_1\backslash\{0\})$. 
    
    Since the partition function of the 2-spin system is zero-free $\mathbf{U}$, 
    for any $G=(V,E)\in\mathcal{G}$, any feasible partial configuration $\sigma_{\Lambda}$, and any vertex $v$ proper to $\sigma_\Lambda$, the functions    $g_{G,v}^{\sigma_{\Lambda}}(\beta, \gamma, \lambda)$ avoid the values $0$ and $1$ on $\mathbf{U}'$. Still by  Montel's theorem (Theorem \ref{thm:montel}), these functions  form a normal family. 
   %Also, all functions in the family take a value in the interval $(0,1)$ at the point $(\beta_0, \gamma_0, \lambda_0)$ for some $\lambda_0\in \mathbb{D}_1$ satisfying $f(\lambda_0)>0$. 
   Then, 
   similar to the proof of Lemma~\ref{bdcircle}, 
   we deduce that for any $\{\beta_1\}\times \{\gamma_1\}\times \partial\mathbb{D}_{\rho}\subseteq {\bf U'}$, we can get a uniform bound for  $g^{\sigma_\Lambda}_{G,v, \beta_1, \gamma_1}(\lambda)=g^{\sigma_\Lambda}_{G,v}(\beta_1, \gamma_1, \lambda)$  on the circle $\partial\mathbb{D}_\rho$.

    Since $f(0)=0$, we have $f(\lambda)=\lambda h(\lambda)$ for some $h(\lambda)$ analytic near $0$.
    Also, for any fixed $\beta\in U_1$ and $\gamma\in U_2$, 
    same as the proof of Lemma~\ref{zftoLDC}, we have 
    \[\begin{aligned}
        g^{\sigma_{\Lambda_1}}_{G,\beta, \gamma,v}(\lambda)-g^{\tau_{\Lambda_2}}_{G, \beta, \gamma,v}(\lambda) 
        & =f(\lambda)^{d_{G}(v,\sigma_{\Lambda_1}\neq\tau_{\Lambda_2})+1} H(f(\lambda)) \\
        & =(\lambda h(\lambda))^{d_{G}(v,\sigma_{\Lambda_1}\neq\tau_{\Lambda_2})+1}\Tilde{H}(\lambda)
    \end{aligned}\]
    where $H(f(\lambda))$ and $\Tilde{H}(\lambda)$ are  analytic near $0$.
    Thus, the coefficients of the term $\lambda^k$ for $0\leq k \leq d_G(u,v)$ in the Taylor series  of $g^{\sigma_{\Lambda_1}}_{G,\beta, \gamma,v}(\lambda)$ and $g^{\tau_{\Lambda_2}}_{G, \beta, \gamma,v}(\lambda)$ near $0$ are the same.
    Then the result follows by Lemma \ref{lem:bound strip} with $P(z)=g^{\sigma_{\Lambda_1}}_{G,\beta, \gamma,v}(z)$ and $Q(z)=g^{\tau_{\Lambda_2}}_{G, \beta, \gamma,v}(z)$. 
\end{proof}

Similarly, we have the following SSM results for $\beta$ and $\gamma$. 

\begin{theorem}[SSM for $\beta$]\label{thm:beta}
   Let $\mathcal{G}$ be
a family of graphs closed under {\rm SAW} tree constructions, $U_2$ and $U_3$ be quasi-neighborhoods of some complex $\gamma_0, \lambda_0\neq 0$ respectively where $0\notin U_2$, and $U_1$ be a neighborhood of $1/U_2$. 
Denote  $(U_1\times U_2 \backslash\{(0, 0)\})\times (U_3\backslash\{0\})$ by $\mathbf U$.
Suppose that the partition function of the 2-spin systems defined on $\mathcal{G}$ is zero-free on $\mathbf{U}$, and one of the two conditions of Lemma~\ref{bdcircle} holds. 
Then, for any $(\beta, \gamma, \lambda) \in \mathbf U$, the corresponding 2-spin system exhibits SSM. 

Moreover, if $\mathcal{G}$ is a family of graphs with bounded degree, then $\gamma_0$ and $\lambda_0$ can be relaxed to be any non-zero complex numbers. 
\end{theorem}

\begin{theorem}[SSM for $\gamma$]\label{thm:gamma}
   Let $\mathcal{G}$ be
a family of graphs closed under {\rm SAW} tree constructions, $U_1$ and $U_3$ be quasi-neighborhoods of some complex $\beta_0, \lambda_0\neq 0$ respectively where $0\notin U_1$, and $U_2$ be a neighborhood of $1/U_1$. 
Denote  $(U_1\times U_2 \backslash\{(0, 0)\})\times (U_3\backslash\{0\})$ by $\mathbf U$. 
Suppose that the partition function of the 2-spin systems defined on $\mathcal{G}$ is zero-free on $\mathbf{U}$, and one of the two conditions of Lemma~\ref{bdcircle} holds.
Then, for any $(\beta, \gamma, \lambda) \in \mathbf U$, the corresponding 2-spin system exhibits SSM. 

Moreover, if $\mathcal{G}$ is a family of graphs with bounded degree, then $\beta_0$ and $\lambda_0$ can be relaxed to be any non-zero complex numbers. 
\end{theorem}

Combining Theorems~\ref{mainany}, \ref{thm:beta} and \ref{thm:gamma}, we get our main Theorem~\ref{thm:main}.

In addition, for the Ising model, we have the following SSM result from zero-free regions of $\beta$.

\begin{theorem}[Theorem~\ref{mainIsing} restated, SSM for $\beta$ in the Ising model]\label{thm:ssmising}
    Let $\mathcal{G}$ be
a family of graphs closed under {\rm SAW} tree constructions, $\lambda_0$ be a nonzero complex number, and $U$ be a neighborhood of $1$ (or $-1$).
%and $U_3$ be a quasi-neighborhood of some $\lambda_0>0$. 
%Denote  $(U_1\backslash\{0\})\times (U_3\backslash\{0\})$ by $\mathbf U$. 
Suppose the partition function %of the Ising model
$Z^{\sigma_\Lambda}_{G,\lambda_0}(\beta)=Z^{\sigma_\Lambda}_{G}(\beta, \beta, \lambda_0)$
defined on $\mathcal{G}$  with the external field $\lambda_0$ is zero-free on $\beta\in {U}$.  
Then, for any $\beta  \in  U$, the corresponding Ising model exhibits SSM. 
\end{theorem}
\begin{proof}
%Fix some complex $\lambda_0\neq 0$. 
For any graph $G\in \mathcal{G}$, any feasible partial configuration $\sigma_\Lambda$ and any proper vertex $v$, 
   % For any complex $\lambda_0\neq 0$ and any rational function $P^{\sigma_\Lambda}_{G,v}$, 
   one can easily check that $$P^{\sigma_\Lambda}_{G,v}(1, 1, \lambda_0)=\frac{\lambda_0}{\lambda_0+1}, \text{ and } P^{\sigma_\Lambda}_{G,v}(-1, -1, \lambda_0)=\frac{\lambda_0}{\lambda_0+1} \text{ or } \frac{-\lambda_0}{-\lambda_0+1} .$$ 
   Note that the partition function $Z^{\sigma_\Lambda}_{G}(\beta, \beta, \lambda_0)$ is zero-free on $\beta\in U$ where $U$ is a neighborhood of $1$ (or $-1$).  
    %Since the 
    Thus,
    for the family $\{P^{\sigma_\Lambda}_{G,v}\}$ of all rational functions, 
    we have $|P^{\sigma_\Lambda}_{G,v}(1, 1, \lambda_0)|<\infty$ if $1\in U$ and $|P^{\sigma_\Lambda}_{G,v}(-1, -1, \lambda_0)|<\infty$ if $-1\in U$. 
    %their values on the inputs $(1, 1, \lambda_0)$ and $(-1, -1, \lambda_0)$ are both bounded. 
    Following the proof of Lemma~\ref{bdcircle}, we can  prove a uniform bound for $\{P^{\sigma_\Lambda}_{G,v}\}$ on any compact set in $U$ 
    %for any complex $\lambda_0\neq 0$ 
    without requiring $\mathcal{G}$ to be a family of graphs of bounded degree. 
    The other parts of the proof are  the same as the proofs of Lemma~\ref{zftossmdisk} and Theorem~\ref{mainany}. 
\end{proof}

\section{The implication for non-uniform external fields}\label{sec:LY}
In this section, 
we extend our result to handle
 2-spin systems with non-uniform external fields, and apply it to get new spatial mixing results for the non-uniform ferromagnetic Ising model from corollaries of  Lee-Yang circle theorem on graphs with pinned vertices.

\begin{definition}\label{def:non-uniform}
    A 2-spin system defined on a graph $G=(V, E)$ is said to has non-uniform external fields if vertices $v\in V$ may have different external fields $\lambda_v$. 
The partition function of such a system is defined to be
\[Z_G(\beta,\gamma,\boldsymbol{\lambda}):=\sum\limits_{\sigma:V\to\{+,-\}}w(\sigma)\]
with $w(\sigma)=\beta^{m_+(\sigma)}\gamma^{m_-(\sigma)}\prod\limits_{v\in V, \sigma(v)=+}\lambda_v$. 
\end{definition}
The notations and definitions for 2-spin systems with a uniform external field   can be extended to the setting with non-uniform external fields as above. In particular, the Christoffel-Darboux type identity (Theorem~\ref{CDgeneral})  is also valid for 2-spin systems with non-uniform external fields by replacing $\lambda^{d(u,v)+1}$ with $\prod_{w\in p_{uv}}\lambda_w$ where $p_{uv}$ is the unique path connecting $u$ and $v$ in a tree. 
The SAW tree construction also works where all vertices in $V_{\text{SAW}}$ mapped from a single vertex $v\in V$ have the same external field as the external field   $\lambda_v$ of the original  $v$, and 
Theorem~\ref{prop:SAW} still {holds} for the multivariate rational functions $P_{G, v}^{\sigma_\Lambda}(\beta, \gamma,  \pmb \lambda)$.
For non-uniform external fields $\pmb \lambda=(\lambda_v)_{v\in V}$, we say they have the same argument $\theta$ if $\arg(\lambda_v)=\theta$ for all $v\in V$. 

\begin{theorem}
%[Theorem~\ref{thm:non-uniform} restated]
\label{lem:mainmulti}
    Fix $\beta, \gamma, \lambda^\ast >0$ and a family of  graphs $\mathcal{G}$ closed under {\rm SAW} tree constructions. 
    Suppose that for any $G\in \mathcal{G}$ and any feasible partial configuration $\sigma_\Lambda$, the partition function $Z^{\sigma_\Lambda}_{G, \beta, \gamma}(\pmb \lambda)$ of the 2-spin system with non-uniform external fields $\pmb \lambda$ having the same argument $\theta$ and satisfying $0<|\lambda_v|<\lambda^*$ for all $v\in V(G)$ is always non-zero.
    Then, the corresponding 2-spin system 
    specified by $\beta, \gamma$ and $\pmb \lambda$ exhibits SSM. 
\end{theorem}
\begin{proof}
   Since the SAW tree construction still hold for  2-spin systems with non-uniform external fields, similar to Lemma~\ref{lem:tree-enough}, we only need to consider  2-spin systems on trees. 

   For a graph $G$, a feasible partial configuration $\sigma_\Lambda$ together with a vertex $v$ proper to it, and some fixed values $\pmb \lambda$ of external fields satisfying $\arg(\lambda_v)=\theta$ for some $\theta$ and $0<|\lambda_v|<\lambda^*$ for all $v\in V(G)$ and $\sup_{v\in V(G)}|\lambda_v|<\lambda^\ast$, we define the function $p_{G,v}^{\sigma_{\Lambda}}(z)=P_{G,v}^{\sigma_{\Lambda}}(z\boldsymbol{\lambda})$ for $z\in \mathbb{D}_{\rho}$ where $\rho=\frac{\lambda^*}{\sup|\lambda_v|}>1.$
   Note that $p_{G,v}^{\sigma_{\Lambda}}(z)$ is a univariate function on the complex variable $z$. 
Since the partition function
$Z^{\sigma_\Lambda}_{G, \beta, \gamma}(\pmb \lambda)$ is zero-free, the function $p_{G,v}^{\sigma_{\Lambda}}(z)$ avoids $0$ and $1$ on $\mathbb{D}_{\rho}$.

Fix the values of $\pmb \lambda$.
Still by Montel's theorem (Theorem \ref{thm:montel}), for any graph $G$, any  feasible partial configuration $\sigma_\Lambda$, and any vertex $v$ proper to it, the functions $p_{G,v}^{\sigma_{\Lambda}}(z)$ form a normal family. 
Also, all functions in the family take a value in the interval $(0,1)$ when $z=e^{-i\theta}$.
Still similar to the proof Lemma~\ref{bdcircle}, we can show that for any  $0<\rho_0<\rho$, there exists a uniform bound for the family of functions  $p_{G,v}^{\sigma_{\Lambda}}(z)$ on the circle $\partial\mathbb{D}_{\rho_0}$.

    Also, by the Christoffel-Darboux type identity (Theorem~\ref{CDgeneral})  for 2-spin systems with non-uniform external fields in which $\lambda^{d(u,v)+1}$ is replace by $\prod_{w\in p_{uv}}\lambda_w$ where $p_{uv}$ is the unique path connecting $u$ and $v$, we have

    \[\begin{aligned}
        p_{G,v}^{\sigma_\Lambda}(z)-p_{G,v}^{\sigma_\Lambda, u^+}(z) & =P_{G,v}^{\sigma_\Lambda}(z\boldsymbol{\lambda})-P_{G,v}^{\sigma_\Lambda, u^+}(z\boldsymbol{\lambda}) \\ & =\frac{Z_{G,u}^{\sigma_\Lambda,+}Z_{G,v}^{\sigma_\Lambda,+}-Z_{G,v,u}^{\sigma_\Lambda,+,+}Z_{G}^{\sigma_\Lambda}}{Z_{G}^{\sigma_\Lambda}Z_{G,u}^{\sigma_\Lambda,+}}(z\boldsymbol{\lambda})\\
        &=\prod_{w\in p_{uv}}(z\lambda_w)H(z)\\
 &=z^{d_G(u,v)+1}H'(z),
    \end{aligned}\]
for some functions $H(z)$ and $H'(z)$ analytic near $0$. 
The same form of equation holds for $P_{G,v}^{\sigma_\Lambda}(z)-P_{G,v}^{\sigma_\Lambda, u^-}(z)$.
Similarly, we can define point-to-point LDC and LDC for functions $p_{G,v}^{\sigma_\Lambda}(z)$ near $z=0$, and show that point-to-point LDC implies LDC. 
Thus, the functions $p_{G,v}^{\sigma_\Lambda}(z)$ satisfy LDC. 
Combine the uniform bound on a circle and LDC, and we get SSM by Lemma~\ref{lem:bound strip}.
\end{proof}

The following is the celebrated Lee-Yang circle theorem. 

\begin{theorem}[Lee-Yang circle theorem]\label{LY}
    Fix  $\beta>1.$ 
    For any graph $G=(V, E)$, the partition function $Z_G(\boldsymbol{\lambda})$ of the ferromagnetic Ising model specified by $\beta$ and non-uniform external fields $\pmb \lambda=(\lambda_v)_{v\in V}$ is nonzero if we have $|\lambda_v|>1$ for all $v\in V$, or $|\lambda_v|<1$ for all $v\in V$. 
\end{theorem}

Note that Lee-Yang circle theorem gives zero-freeness  for the partition function of the Ising model regarding the external fields $\pmb\lambda$. 
 Thus, we need to apply Theorem~\ref{lem:mainmulti}, instead of Theorem~\ref{thm:ssmising} which requires zero-freeness regarding the edge activity $\beta$\footnote{Zeros of the partition function of the Ising model regarding $\beta$ are  known as Fisher zeros \cite{fisher1965nature}.}. 
 However, Theorem~\ref{lem:mainmulti} requires the zero-freeness condition holds for 2-spin systems on graphs even conditioning on some partial configurations.
Thus, Lee-Yang circle theorem cannot be directly adapted to our setting since it holds only for graphs with no pinned vertices. 
In order to obtain SSM from zero-freeness, 
we first give the following corollary
of Lee-Yang circle theorem.

\begin{corollary}\label{cor:LY}
    Fix  $\beta>1$ and integer $d\geq 2$. 
    For any graph $G=(V, E)$ of degree at most $d$ and any feasible partial configuration $\sigma_\Lambda$, the partition function $Z^{\sigma_\Lambda}_G(\boldsymbol{\lambda})$ 
    %of the ferromagnetic Ising model specified by $\beta$ and non-uniform external fields $\pmb \lambda=(\lambda_v)_{v\in V}$ 
    is nonzero if we have $|\lambda_v|>\beta^d$ for all $v\in V$, or $0<|\lambda_v|<1/\beta^d$ for all $v\in V$. 
\end{corollary}
\begin{proof}
Consider the partition function $Z^+_{G, v}(\beta, \pmb \lambda)$ of the Ising model with non-uniform external fields on a graph $G=(V, E)$ with a vertex $v$ pinned to $+$. Note that
\begin{equation*}
\begin{aligned}
Z^+_{G, v}(\beta, \pmb \lambda)=\sum_{\sigma:V\backslash\{v\}\rightarrow \{+, -\}}\beta^{m(\sigma)}\lambda_v\prod\limits_{w\in V\backslash\{v\} \atop \sigma(w)=+}\lambda_w,
\end{aligned}
\end{equation*}
where $m(\sigma)$ is the number of $(+,+)$ and $(-,-)$ edges in $G$ under the configuration $\sigma$ on $V\backslash \{v\}$ together with $v$ pinned to $+$. Let $m'(\sigma)$ be the number of $(+,+)$ and $(-,-)$ edges in $G\backslash\{v\}$ under the configuration $\sigma$ on $V\backslash \{v\}$.
Then, $m(\sigma)=m'(\sigma)+|\{w\in N(v)\mid \sigma(w)=+\}|$ where $N(v)$ is the set of neighbors of $v$. Thus, we have 
\begin{equation*}
\begin{aligned}
Z^+_{G, v}(\beta, \pmb \lambda)&=\lambda_v\left(\sum_{\sigma:V\backslash\{v\}\rightarrow \{+, -\}}\beta^{m'(\sigma)+|\{w\in N(v)\mid \sigma(w)=+\}|}\prod\limits_{w\in V\backslash\{v\} \atop \sigma(w)=+}\lambda_w\right)\\
&=\lambda_v\left(\sum_{\sigma:V\backslash\{v\}\rightarrow \{+, -\}}\beta^{m'(\sigma)}\prod\limits_{w\in V\backslash (N(v)\cup\{v\}) \atop \sigma(w)=+}\lambda_w\prod\limits_{w\in N(v) \atop \sigma(w)=+}\beta\lambda_w\right)\\
&=\lambda_v  Z_{G\backslash\{v\}}(\beta, \pmb \lambda^{v^+}) 
\end{aligned}
\end{equation*}
where $Z_{G\backslash\{v\}}(\beta, \pmb \lambda^{v^+})$ is the partition function of the Ising model with non-uniform external fields $\pmb \lambda^{v^+}$ on the graph $G\backslash\{v\}$ obtained from $G$ by deleting $v$, and $\lambda^{v^+}_w=\lambda_w$ for $w\in V\backslash (N(v)\cup\{v\})$ and $\lambda^{v^+}_w=\beta\lambda_w$ for $w\in N(v)$. 
Similarly, we have $$Z^-_{G, v}(\beta, \pmb \lambda)= \beta^{|N(v)|}Z_{G\backslash\{v\}}(\beta, \pmb \lambda^{v^-})$$ where $\lambda^{v^-}_w=\lambda_w$ for $w\in V\backslash (N(v)\cup\{v\})$ and $\lambda^{v^-}_w=\lambda_w/\beta$ for $w\in N(v)$.

Then, for the partition function $Z_G^{\sigma_\Lambda}(\beta, \pmb \lambda)$ of the Ising model on a graph $G=(V, E)$ conditioning on a partial configuration $\sigma_\Lambda$, we can delete all vertices in $\Lambda$ from $G$, and  for every $w\in V\backslash \Lambda$ multiply the external field $\lambda_w$ by $\beta^{n_+(w)-n_-(w)}$  where $n_+(w)$ is the number of neighbors of $w\in \Lambda$  that are pinned to $+$ under $\sigma_\Lambda$ and $n_-(w)$ is the number of neighbors of $w\in \Lambda$  that are pinned to $-$ under $\sigma_\Lambda$.
We denote $\beta^{n_+(w)-n_-(w)}\lambda_w$ by $\lambda_w^{\sigma_\Lambda}$ and $(\lambda_w^{\sigma_\Lambda})_{w\in V\backslash \Lambda}$ by $\pmb \lambda^{\sigma_\Lambda}$.
Then, we have 
$$Z_G^{\sigma_\Lambda}(\beta, \pmb \lambda)=\beta^{m_1+m_2}Z_{G\backslash\Lambda}(\beta, \pmb \lambda^{\sigma_\Lambda})\prod_{w\in \Lambda \atop \sigma_\Lambda(w)=+}\lambda_w$$
where $m_1$ is the number of $(+, +)$ and $(-, -)$ edges in the subgraph of $G$ induced by $\Lambda$ conditioning on $\sigma_\Lambda$ and $m_2=\sum_{w\in \Lambda \atop \sigma_{\Lambda}(w)=-}|N(w)\cap (V\backslash \Lambda)|$.
Note that $\beta>1$ and $\lambda_v\neq 0$ for all $v\in V$. 
Also since the degree of vertices in $G$ is at most $d$, for every $w\in V\backslash\Lambda$, we have $$|\lambda_w^{\sigma_\Lambda}|=\beta^{n_+(w)-n_-(w)}|\lambda_w|\in [\beta^{-d}|\lambda_w|, \beta^d|\lambda_w|].$$
Clearly, $|\lambda_w^{\sigma_\Lambda}|>1$ if $|\lambda_w|> \beta^d$ and $|\lambda_w^{\sigma_\Lambda}|<1$ if $|\lambda_w|< \beta^{-d}$.
If for all $v\in V$,  $|\lambda_v|>\beta^d$ or $0<|\lambda_v|<1/\beta^d$, then $Z_{G\backslash\Lambda}(\beta, \pmb \lambda^{\sigma_\Lambda})\neq 0$ by Theorem~\ref{LY}. 
Thus, $Z_G^{\sigma_\Lambda}(\beta, \pmb \lambda)\neq 0$. 
\end{proof}

Combining Theorem~\ref{lem:mainmulti} and Corollary~\ref{cor:LY}, we get a new SSM result for the ferromagnetic Ising model with non-uniform external fields. 

\begin{corollary}
%[Corollary~\ref{ourIsing} restated]
    Fix $\beta>1$ and integer $d\ge 2$. 
    %Suppose that $\lambda\in \mathbb{C}$ and $|\lambda|$
    The ferromagnetic Ising model on graphs $G$ of  degree at most $d$ specified by $\beta$ and non-uniform positive external fields $\pmb \lambda$ with the same argument $\theta$  exhibits SSM if  $\inf_{v\in V(G)} |\lambda_v|>\beta^d$ or $\sup_{v\in V(G)}|\lambda_v|<1/\beta^d$.
     %In particular, for  a  complex uniform $\lambda$, SSM holds if $|\lambda|>\beta^d$ or $0<|\lambda|<1/\beta^d$.   
\end{corollary}
\begin{proof}
    %By spin reversal symmetry (Remark \ref{Isingsym}), it suffices to 
    We first consider the case $\sup_{v\in V(G)}\lambda_v<1/\beta^d$. 
    In this case, zero-freeness is guaranteed by Corollary \ref{cor:LY}, and then SSM follows by Theorem \ref{lem:mainmulti}. 

     For the case that $\inf_{v\in V(G)}\lambda_v>\beta^d$, notice that by switching the meanings of $+$ and $-$ spins,
     the \emph{spin reversal symmetry} holds:
     For any graph $G=(V, E)$ and 
any partial configuration $\sigma_\Lambda$,
$$Z_G^{\sigma_\Lambda}(\beta, \gamma,\pmb \lambda)=Z_G^{\overline{\sigma}_\Lambda}(\gamma, \beta, \pmb \lambda^{-1})\prod_{v\in V}\lambda_v $$
where $\pmb \lambda^{-1}=(\lambda_v^{-1})_{v\in V}$ and $\overline{\sigma}_\Lambda$ is the partial configuration on $\Lambda$ such that $\overline{\sigma}_\Lambda(w)=\overline{\sigma_\Lambda(w)}$ (i.e., the spin $\sigma_\Lambda(w)$ imposed to $w$ is flipped). 
   Thus, the Ising model with non-uniform external fields $\pmb\lambda$ can be transferred to the Ising model with $\pmb\lambda^{-1}$. 
   Since $\inf_{v\in V(G)}\lambda_v>\beta^d$, we have $0<\sup_{v\in V(G)}\lambda^{-1}_v<1/\beta^d$. 
    Thus, SSM also holds for the case that $\inf_{v\in V(G)}\lambda_v>\beta^d$.
%    Clearly, for a uniform $\lambda$, the condition $\inf_{v\in V(G)} |\lambda_v|>\beta^d$ or $\sup_{v\in V(G)}|\lambda_v|<1/\beta^d$ is equivalent to $|\lambda|>\beta^d$ or $0<|\lambda|<1/\beta^d$.
\end{proof}

If we restrict the partial configuration $\sigma_\Lambda$ to include at most one vertex with the $+$ spin (or at most one vertex with the $-$ spin), then similar to Corollary~\ref{cor:LY}, we have the following zero-free results as a corollary of Lee-Yang Theorem.

\begin{corollary}\label{cor:LY2}
    Fix $\beta>1$. 
    For any graph $G=(V, E)$ and any feasible partial configuration $\sigma_\Lambda$ with at most one vertex $u\in \Lambda$ having $\sigma(u)=+$, the partition function $Z^{\sigma_\Lambda}_G(\boldsymbol{\lambda})$ 
    %of the ferromagnetic Ising model specified by $\beta$ and non-uniform external fields $\pmb \lambda=(\lambda_v)_{v\in V}$ 
    is nonzero if  for all $v\in V$,  $0<|\lambda_v|<1/\beta$.
    
   Symmetrically, for any graph $G$ and any feasible partial configuration $\sigma_\Lambda$ with at most one vertex $u\in \Lambda$ having $\sigma(u)=-$, the partition function $Z^{\sigma_\Lambda}_G(\boldsymbol{\lambda})$ 
    %of the ferromagnetic Ising model specified by $\beta$ and non-uniform external fields $\pmb \lambda=(\lambda_v)_{v\in V}$ 
    is nonzero  if  for all $v\in V$,  $|\lambda_v|>\beta$. 
\end{corollary}

We say a partial configuration $\sigma_\Lambda$ is all-plus 
%(resp. all-plus), denoted by $\sigma^-_\Lambda$ (resp. $\sigma^+_\Lambda$), 
 if $\sigma(v)=+$ for every $v\in \Lambda$, and all-minus if $\sigma(v)=-$ for every $v\in \Lambda$.
%A partial configuration is said to be pure if it is all-minus or all-plus. 
%We will use $\sigma^{\pm}_\Lambda$ to denote a pure partial configuration. 
%Let 
%$Q_{G,v}^{\sigma_\Lambda}(\beta, \pmb \lambda)=\frac{Z^{\sigma_\Lambda,-}_{G,v}(\beta,\pmb \lambda)}{Z^{\sigma_\Lambda}_{G}(\beta,\pmb\lambda)}=1-P_{G,v}^{\sigma_\Lambda}(\beta, \pmb \lambda)$. 
%Then, we have the following uniform bound. 
%Then, we have the following uniform bound result.  
%\begin{lemma}
%    Fix $\beta>1$ and integer $d\geq 0$. 
% For all $\lambda^*_\ge\lambda^*_2\ge \beta$ and forall $\lambda_*\le 1/\beta$, there exists an $M>0$ such that for every $G$ of degree at most $d$, any feasible all-minus partial configuration $\sigma^-_\Lambda$, and any vertex $v$ proper to $\sigma^-_\Lambda$,  we have $|Q^{\sigma^{-}_\Lambda}_{G,v}(\lambda)|\leq M$ for all $\pmb \lambda$ with $\lambda_1^*\ge \sup_{v\in V(G)}|\lambda_v| \ge \inf_{v\in V(G)}  |\lambda_v|\ge \lambda^*_2>\beta$ or $\sup_{v\in V(G)}|\lambda_v|\le\lambda_*<1/\beta$. 
%\end{lemma}
%\begin{proof}
%    By corollary~\ref{cor:LY2},  $Q^{\sigma^{-}_\Lambda}_{G,v}(\lambda)$ always avoids $0$ and $1$ when $\inf_{v\in V(G)} |\lambda_v|>\beta$ or $\sup_{v\in V(G)}|\lambda_v|<1/\beta$. \textcolor{red}{to Xiaowei, please finish a short proof here.}
%\end{proof}
We define the following forms of spatial mixing.

\begin{definition}[Plus or minus spatial mixing (PSM/MSM)]\label{def:minus-mixing}
    Fix parameters $\beta, \gamma$, non-uniform $\pmb \lambda$, and %where %$(\beta, \gamma)\neq (0, 0)$ and $\lambda\neq 0$, and 
    a family of graphs $\mathcal{G}$. 
The  corresponding 
2-spin system 
 defined on  $\mathcal{G}$ specified by $\beta, \gamma$ and non-uniform $\pmb{\lambda}$ is said to satisfy \emph{plus spatial mixing (PSM) (or respectively minus spatial mixing (MSM))} with exponential rate $r>1$ if there exists a constant $C$ such that for any $G=(V,E)\in\mathcal{G}$,
any feasible all-plus (or respectively all-minus) partial configurations $\sigma^{}_{\Lambda_1}$ and $\tau^{}_{\Lambda_2}$ where $\Lambda_1$ may be different with $\Lambda_2$, and any vertex $v$ proper to $\sigma^{}_{\Lambda_1}$ and $\tau^{}_{\Lambda_2}$, we have 
\[\left|P_{G,v}^{\sigma^{}_{\Lambda_1}}-P_{G,v}^{\tau^{}_{\Lambda_2}}\right|\leq Cr^{-d_G(v,\sigma^{}_{\Lambda_1}\neq \tau^{}_{\Lambda_2})}.\]
\end{definition}

\begin{theorem}\label{thm:msm}
    Fix $\beta>1$ and integer $d\ge 2$. 
    %Suppose that $\lambda\in \mathbb{C}$ and $|\lambda|$
    The ferromagnetic Ising model on all graphs $G$  specified by $\beta$ and non-uniform positive external fields $\pmb \lambda$ with the same argument $\theta$ exhibits PSM if $\inf_{v\in V(G)} |\lambda_v|>\beta$, and it exhibits MSM if $\sup_{v\in V(G)}|\lambda_v|<1/\beta$. 
       
\end{theorem}
\begin{proof}
    By symmetry, it suffices to prove MSM when $\sup_{v\in V(G)}|\lambda_v|<1/\beta$. 
    Consider the family of functions $\{P_{G,v}^{\tau_\Lambda}(z\pmb{\lambda})\}$ defined on $U=\{z\in \mathbb{C}\mid |z|<\frac{1}{\beta \sup_{v\in V(G)}|\lambda_v|}\}$, where $G$ runs over all graphs, $\tau_\Lambda$ runs over all all-minus partial configurations, and $v$ runs over all proper vertices. 

    Note that for any $z\in U$ and any $\lambda_v$, we have $|z\lambda_v|<1$. %any $\lambda_v$ with $\sup_{v\in V(G)}|\lambda_v|<1/\beta$, $|z\lambda_v|$.
   Then, by Corollary \ref{cor:LY2}, %$Q_{G,v}^{\tau_\Lambda^{-}}(z\pmb{\lambda})$ and 
   the function $P_{G,v}^{\tau_\Lambda}(z\pmb{\lambda})=\frac{Z_{G,v}^{\tau_\Lambda, +}(z\pmb{\lambda})}{Z_{G}^{\tau_\Lambda}(z\pmb{\lambda})}$ is nonzero since the configuration $\tau_\Lambda$ together with $v$ pinned to $+$ has exactly one vertex $v$ with the $+$ spin. 
Also, $P_{G,v}^{\tau_\Lambda}(z\pmb{\lambda})=1-\frac{Z_{G,v}^{\tau_\Lambda, -}(z\pmb{\lambda})}{Z_{G}^{\tau_\Lambda}(z\pmb{\lambda})}\neq 1$ since $Z_{G,v}^{\tau_\Lambda, -}(z\pmb{\lambda})\neq 0$.

Thus, all the functions $\{P_{G,v}^{\tau_\Lambda}(z\pmb{\lambda})\}$ avoid $0$ and $1$. They form a normal family by Montel's Theorem (Theorem \ref{thm:montel}). 
    Hence, there exists a uniform bound $M>0$ such that 
    $$|P_{G,v}^{\tau_\Lambda}(z\pmb{\lambda})|\le M, \quad \forall |z|\le \rho:=\frac{1}{2}\left(1+\frac{1}{\beta \sup_{v\in V(G)}|\lambda_v|}\right).$$
    Combining with the LDC property proved above, we can prove MSM for $P_{G,v}^{\tau_\Lambda}(z\pmb{\lambda})$ with $|z|\leq \rho$. 
    Since $\sup_{v\in V(G)}|\lambda_v|<1/\beta$, one can check that $\rho > 1$. 
    Thus,  we can take $z=1$. We are done with the proof. 
    %and the result follows by Lemma \ref{lem:bound strip}. 
\end{proof}

\section{Concluding remarks and questions}\label{sec:conclude}

We prove that if the partition function of a 2-spin system is zero-free in a complex neighborhood of $\lambda=0$ or $\beta\gamma=1$ containing a positive point, then the corresponding 2-spin system exhibits SSM.
By a Christoffel-Darboux type identity, it seems that the implication from zero-free regions to SSM 
only works for  complex neighborhoods of $\lambda=0$ or $\beta\gamma=1$. 
A natural question is: \emph{what is the magic of these two special settings of parameters $\lambda=0$ and $\beta\gamma=1$ for which the exact computation of the partition function is polynomial-time solvable?}
In fact, so far all currently known zero-free regions (and even all SSM results and all FPTAS results) for  2-spin systems lie in some complex neighborhoods of $\lambda=0$ (up to a switch of the spins $+$ and $-$) or $\beta\gamma =1$. 
We are curious whether this is always true and we make the following conjectures. 

\begin{conjecture}
   Fix $\beta$ and $\gamma$ that are not both zero.
   Suppose that the partition function $Z^{\sigma_\Lambda}_{G}(\lambda)$ of the 2-spin system is zero-free on a complex region $U$. 
   Then, there exists a complex neighborhood $U'$ of $0$ such that $U\subseteq U'$ and $Z^{\sigma_\Lambda}_{G}(\lambda)$ is zero-free on $U'$. 
\end{conjecture}

\begin{conjecture}
   Fix $\gamma\neq 0$ and $\lambda\neq 0$. 
   Suppose that the partition function $Z^{\sigma_\Lambda}_{G}(\beta)$ of the 2-spin system is zero-free on a complex region $U$. 
   Then, there exists a complex neighborhood $U'$ of $1/\gamma$ such that $U\subseteq U'$ and $Z^{\sigma_\Lambda}_{G}(\beta)$ is zero-free on $U'$. 
\end{conjecture}

\begin{conjecture}\label{conjecture3}
   Fix $\lambda\neq 0$. 
   Suppose that the partition function $Z^{\sigma_\Lambda}_{G}(\beta)$ of the Ising model is zero-free on a complex region $U$. 
   Then, there exists a complex neighborhood $U'$ of $1$ (or $-1$) such that $U\subseteq U'$ and $Z^{\sigma_\Lambda}_{G}(\beta)$ is zero-free on $U'$. 
\end{conjecture}

For the Ising model, a complex zero-free neighborhood of $\beta=1$ is known \cite{LSSFisherzeros}. 
Also, one can check that the point $\beta=-1$ (with a fixed $\lambda$) satisfies the real contraction property introduced in \cite{shao2019contraction} which implies that there exists a complex neighborhood of $-1$ such that the partition function of the Ising model is zero-free on it.
Thus, for Conjecture \ref{conjecture3}, at least there exists a complex neighborhood of $\beta=1$ and a complex neighborhood of $\beta=-1$ respectively such that the partition function of the Ising model is zero-free on them.

\appendix
\section{Christoffel-Darboux identity for \texorpdfstring{$q$}{q}-spin systems on trees}\label{sec:CDqspin}

In this section, we extend Christoffel-Darboux type identity to $q$-spin systems.

Such a system is defined on a simple undirected graph $G=(V,E)$  and in this way the individual entities comprising the system correspond to the
vertices $V$ and their pairwise interactions correspond to the edges $E$.

The system is associated with some parameters consisting of a $q\times q$ symmetric matrix of edge activities $A=(a_{i,j})_{1\le i,j\le q}\in\mathcal{M}_q(\mathbb{C})$ that models the tendency of vertices to agree and disagree with their neighbors, and a uniform vertex activity vector $\boldsymbol{\lambda}=(\lambda_1,\cdots,\lambda_q)\in\mathbb{C}^q$ that models an external field which determines the propensities of a vertex. 

A partial configuration of this system refers to a mapping $\sigma: \Lambda\to[q]=\{1,2,\cdots,q\}$ for some $\Lambda\subseteq V$ which may be empty. It assigns one of the $q$ spins to each vertex in $\Lambda$. 
When $\Lambda=V$, it is called a configuration and its weight is 
\[\omega(\sigma)=\prod_{v\in V}\lambda_{\sigma(v)}\prod_{(u,v)\in E}a_{\sigma(u),\sigma(v)}.\] 
The \emph{partition function} of a 2-spin system is defined to be
\[Z_G(A,\boldsymbol{\lambda}):=\sum\limits_{\sigma:V\to[q]}w(\sigma).\] 
We also define the partition function   conditioning on a pre-described partial configuration $\sigma_\Lambda$ by \[Z^{\sigma_\Lambda}_G(A,\boldsymbol{\lambda})\sum_{\tau:V\to[q],\tau|_{\Lambda}=\sigma_\Lambda}w(\sigma)\] 
where $\tau|_{\Lambda}$ denotes the restriction of the configuration $\tau$ on $\Lambda$. For $u,v\in V(G),$ we define
\[Z^{\sigma_\Lambda,i}_{G,v}(A,\boldsymbol{\lambda})=\sum\limits_{\sigma|_{\Lambda}=\sigma_\Lambda,\sigma(v)=i}w(\sigma) \text{ and } Z^{\sigma_\Lambda,i,j}_{G,v,u}(A,\boldsymbol{\lambda})=\sum\limits_{\sigma|_{\Lambda}=\sigma_\Lambda,\sigma(v)=i,\sigma(u)=j}w(\sigma).\]
If $v\in\Lambda$ is already pinned to some $j\in[q]$, then automatically $Z^{\sigma_\Lambda,i}_{G,v}=0$ for all $i\neq j$ and $Z^{\sigma_\Lambda,j}_{G,v}=Z^{\sigma_\Lambda}_{G}.$ 
Given a partial configuration $\sigma_\Lambda$,  for any $\Lambda'\subset\Lambda,$ we denote $\sigma_{\Lambda'}$ the restriction of $\sigma_{\Lambda}$ on $\Lambda'$, called a  sub-partial configuration of $\sigma_{\Lambda}$. In the rest of this paper, we omit the argument $(A,\lambda)$ of the partition functions for simplicity. 

Now we generalize Theorem \ref{CDgeneralintro} to $q$-spin systems as the following: 
\begin{theorem}\label{thm:mainmul}\label{thm:CD-q-spin}
Suppose that $T$ is a tree, $\sigma_{\Lambda}$ is a partial configuration on some $\Lambda \subseteq V$ and $u\neq v$ are two vertices in $V\backslash \Lambda$.
    Let $P$ denote  the unique path in $T$ connecting $u$ and $v$, $N[P]$ denote the set of neighbors of points in $P$,  $N[P]\setminus P=\{v_1,\cdots,v_n\}$, and $T_s$ denote the component of $v_s$ in $T\setminus P$. 
    Then for the partition function of the $q$-spin system defined on $T$ conditioning on $\sigma_{\Lambda}$, we have 
\[\det(Z_{T,u,v}^{\sigma_{\Lambda},i,j})_{1\le i,j\le q}=\begin{cases}(\det A)^{d(u,v)}\prod\limits_{i=1}^{q}\lambda_s^{d(u,v)+1}\prod\limits_{s=1}^{n}\prod\limits_{t=1}^{q}\left(\sum\limits_{k=1}^{q}a_{t,k}Z_{T_s,v_s}^{\sigma_{\Lambda},k}\right) & , \text{if } P\cap\Lambda=\varnothing\\
0 & , \text{if } P\cap \Lambda\neq\varnothing
\end{cases}.\]
\end{theorem}

\begin{proof}
    We prove the equality by induction on $D=d(u,v)$ when $P\cap\Lambda=\varnothing$. When $k=1,$ $u$ and $v$ are adjacent, Denote $u,v_1,\cdots,v_l$ neighbors of $v$ and $v,v_{l+1},\cdots,v_{n}$ neighbors of $u$, as illustrated in the following figure, then
    \[Z_{T,u,v}^{\sigma_{\Lambda},i,j}=\lambda_i\lambda_j a_{i,j}\prod\limits_{s=1}^{l}\left(\sum\limits_{k=1}^{q}a_{k,j}Z_{T_s,v_s}^{\sigma_{\Lambda},k}\right)\prod\limits_{s=l+1}^{n}\left(\sum\limits_{k=1}^{q}a_{i,k}Z_{T_s,v_s}^{\sigma_{\Lambda},k}\right),\]
thus we deduce that
\[\begin{aligned}
\det(Z_{T,u,v}^{\sigma_{\Lambda},i,j})_{1\le i,j\le q}
& = \sum\limits_{\tau\in\mathcal{S}_q}\text{sgn}(\tau)\prod\limits_{t=1}^{q}Z_{T,u,v}^{\sigma_{\Lambda},t,\tau(t)} \\
& = \sum\limits_{\tau\in\mathcal{S}_q}\text{sgn}(\tau)\prod\limits_{t=1}^{q}\lambda_t\lambda_{\tau(t)} a_{t,\tau(t)}\prod\limits_{s=1}^{l}\left(\sum\limits_{k=1}^{q}a_{k,\tau(t)}Z_{T_s,v_s}^{\sigma_{\Lambda},k}\right)\prod\limits_{s=l+1}^{n}\left(\sum\limits_{k=1}^{q}a_{t,k}Z_{T_s,v_s}^{\sigma_{\Lambda},k}\right) \\
& = \sum\limits_{\tau\in\mathcal{S}_q}\text{sgn}(\tau)\prod\limits_{t=1}^{q}a_{t,\tau(t)}\prod\limits_{t=1}^{q}\lambda_t^2\prod\limits_{s=1}^{n}\prod\limits_{t=1}^{q}\left(\sum\limits_{k=1}^{q}a_{k,t}Z_{T_s,v_s}^{\sigma_{\Lambda},k}\right)\\
& = \det A \prod\limits_{t=1}^{q}\lambda_t^2\prod\limits_{s=1}^{n}\prod\limits_{t=1}^{q}\left(\sum\limits_{k=1}^{q}a_{k,t}Z_{T_s,v_s}^{\sigma_{\Lambda},k}\right),
\end{aligned}\]
where $\mathcal{S}_q$ denotes the symmetry group of $[q]$ (i.e., the group of all bijections from $[q]$ to itself) and the third equality uses the fact that $A$ is symmetric. 
\begin{figure}[!hbtp]
\centering
	\includegraphics[scale=0.18]{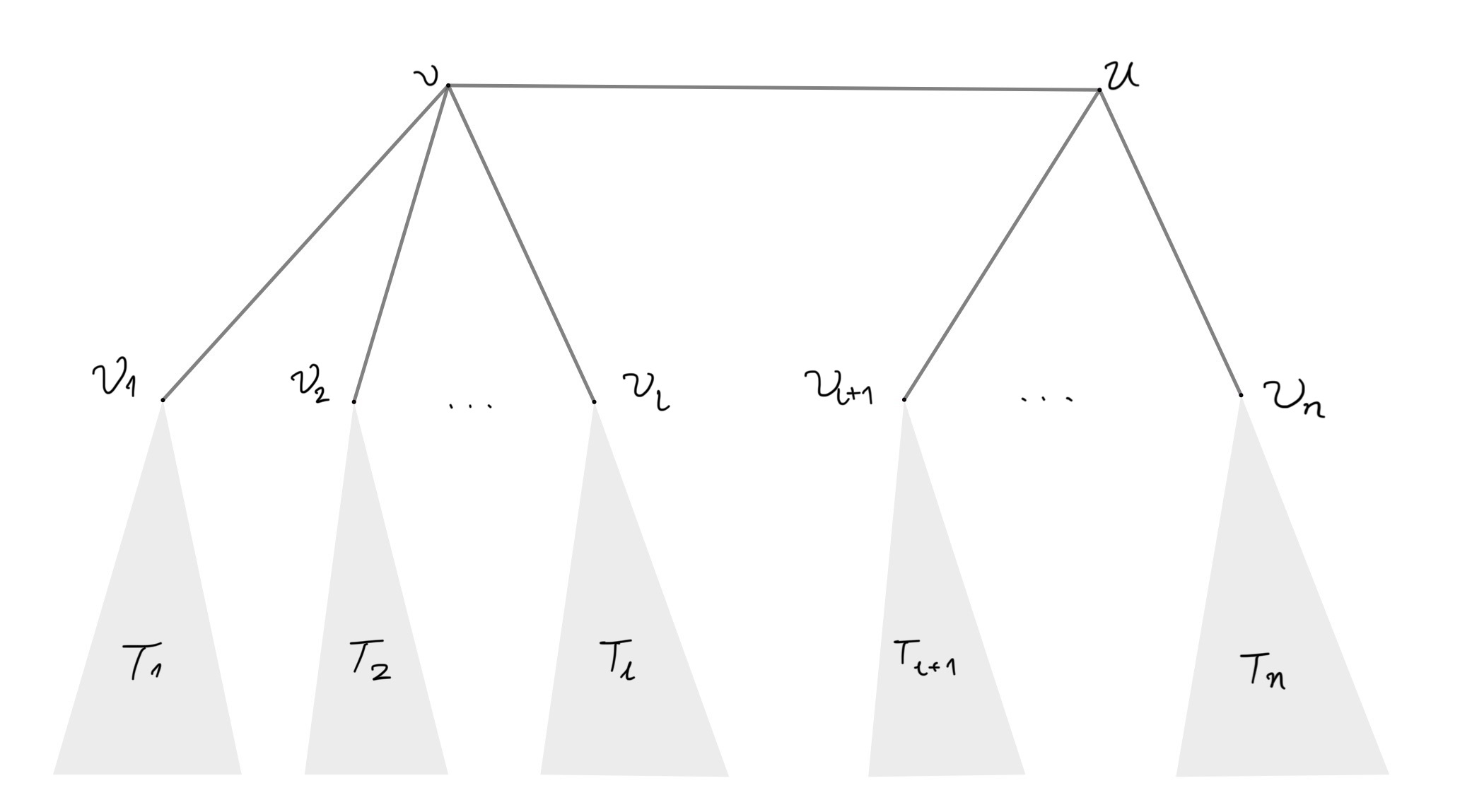}
\end{figure}

Suppose now that $D\ge 2$ and the assertion is true for $D-1$, suppose that $v_0,v_1,\cdots,v_l$ are neighbors of $v$ with $v_0$ on $P.$ Let $T_0$ denote the component of $T\setminus\{v\}$ containing $v_0$ and $u$, then 
\[Z_{T,u,v}^{\sigma_{\Lambda},i,j}=\lambda_j\left(\sum\limits_{k=1}^{q}a_{k,j}Z_{T_0,v_0,u}^{\sigma_{\Lambda},k,i}\right)\prod\limits_{s=1}^{l}\left(\sum\limits_{k=1}^{q}a_{k,j}Z_{T_s,v_s,u}^{\sigma_{\Lambda},k,i}\right),\]
hence
\[\begin{aligned}
\det(Z_{T,u,v}^{\sigma_{\Lambda},i,j})_{1\le i,j\le q}
& = \sum\limits_{\tau\in\mathcal{S}_q}\text{sgn}(\tau)\prod\limits_{t=1}^{q}\lambda_{\tau(t)}\left(\sum\limits_{k=1}^{q}a_{k,\tau(t)}Z_{T_0,v_0,u}^{\sigma_{\Lambda},k,t}\right)\prod\limits_{s=1}^{l}\left(\sum\limits_{k=1}^{q}a_{k,\tau(t)}Z_{T_s,v_s}^{\sigma_{\Lambda},k}\right) \\
& = \prod\limits_{t=1}^{q}\lambda_{t}\prod\limits_{s=1}^{l}\prod\limits_{t=1}^{q}\left(\sum\limits_{k=1}^{q}a_{k,t}Z_{T_s,v_s}^{\sigma_{\Lambda},k}\right)\sum\limits_{\tau\in\mathcal{S}_q}\text{sgn}(\tau)\prod\limits_{t=1}^{q}\left(\sum\limits_{k=1}^{q}a_{k,\tau(t)}Z_{T_0,v_0,u}^{\sigma_{\Lambda},k,t}\right),
\end{aligned}\]
so we are going to study the term
\[\sum\limits_{\tau\in\mathcal{S}_q}\text{sgn}(\tau)\prod\limits_{t=1}^{q}\left(\sum\limits_{k=1}^{q}a_{k,\tau(t)}Z_{T_0,v_0,u}^{\sigma_{\Lambda},k,t}\right),\]
it can be viewed as the determinant of the product of matrices $A$ and $(Z_{T_0,v_0,u}^{\sigma_{\Lambda},k,t})_{1\le k,t\le q}$, thus it equals to the product of their determinants. Now we have
\[\det(Z_{T,u,v}^{\sigma_{\Lambda},i,j})_{1\le i,j\le q}=\prod\limits_{t=1}^{q}\lambda_{t}\prod\limits_{s=1}^{l}\prod\limits_{t=1}^{q}\left(\sum\limits_{k=1}^{q}a_{k,t}Z_{T_s,v_s}^{\sigma_{\Lambda},k}\right)\det A\det (Z_{T_0,v_0,u}^{\sigma_{\Lambda},k,t})_{1\le k,t\le q},\]
notice that $d(u,v_0)=D-1$, we can apply the hypothesis of induction. 

When $P\cap\Lambda\neq\varnothing$, we argue by induction on $l=d_P(u,P\cap\Lambda)+d_P(v,P\cap\Lambda)\ge 2.$ The basic case $l=2$ is simple and the process of induction is exactly as above. 
\end{proof}

\section{Comparison with the cluster expansion approach}
\label{app:comparison}

In \cite{Guus2021zerofreetossm}, 
%Regts established that %absence of complex 
%zero-freeness of the hard-core partition function implies %strong spatial mixing (
%SSM on families of bounded degree graphs. 
%There, the 
the \emph{local dependence of coefficients} (LDC) property is proved using  \emph{cluster expansions} for the hard-core model. 
We first review the cluster expansion formulation.
Then, we discuss its limitations in extending to general 2-spin systems, and explain why the Christoffel-Darboux type identity introduced in the present paper can circumvent the obstacle.

\subsection{Review of the cluster expansion approach}

For the hard-core model on a graph \(G=(V,E)\) with fugacity \(\lambda\), the ratio 
\[
P_{G,v}(\lambda)=\frac{\lambda Z_{G\setminus N[v]}(\lambda)}{Z_G(\lambda)}
\]
(which for positive \(\lambda\) equals the marginal probability that vertex \(v\) is occupied) admits a series expansion around \(\lambda=0\) derived from the cluster expansion of the multivariate independence polynomial. 
Concretely, viewing \(Z_G(\lambda)\) as the specialization of the multivariate polynomial 
\(Z_G(\mathbf{w})=\sum_{I\subseteq V\atop I\text{ indep.}}\prod_{u\in I} w_u\) at \(w_u\equiv \lambda\), one obtains near \(\lambda=0\)
\begin{equation}\label{eq:cluster-expansion}
P_{G,v}(\lambda)=\sum_{k\geq 1}\frac{1}{k!}\sum_{v_1,\dots,v_k\in V}
\phi\bigl(G(v_1,\dots,v_k)\bigr)\,m_v(v_1,\dots,v_k)\,\lambda^k,
\end{equation}
where \(G(v_1,\dots,v_k)\) is the cluster induced by the (not necessarily distinct) vertices \(v_1,\dots,v_k\), 
\(\phi(H)\) is the Ursell function of a graph \(H\), and 
\(m_v(v_1,\dots,v_k)\) counts the occurrences of \(v\) among \(v_1,\dots,v_k\).
Because \(\phi(H)=0\) unless \(H\) is connected, the coefficient of \(\lambda^k\) depends only on the \((k-1)\)-neighborhood of \(v\) in \(G\). 
This immediately yields the LDC property for the hard-core model.

For graph homomorphism measures, a similar cluster expansion is obtained by realizing the partition function as the independence polynomial of an auxiliary graph \(\Gamma\) whose vertices correspond to connected subgraphs of the original graph. 

Let \(G=(V,E)\) be a graph and let \(A\) be a symmetric \(q\times q\) matrix. 
We consider the partition function with an external field \(\xi\in\mathbb{C}^{V\times[q]}\):
\[
Z_G(A,\xi)=\sum_{\phi:V\to[q]}\prod_{v\in V}\xi_{v,\phi(v)}\prod_{uv\in E}A_{\phi(u),\phi(v)}.
\]
Fix a boundary condition \(\sigma:\Lambda\to[q]\) on some \(\Lambda\subseteq V\setminus\{v\}\). 
The goal is to expand the ratio
\[
P_{G,v,i;A}^{\sigma}(z)=\frac{Z_{G,v,i}^{\sigma_{v,i}}(J+z(A-J))}{Z_G^{\sigma}(J+z(A-J))}
\]
around \(z=0\), where \(J\) is the all-ones matrix.

\paragraph{Auxiliary graph \(\Gamma\).}
The vertex set of \(\Gamma\) consists of all connected subgraphs \(H=(S,F)\) of \(G\) with at least one edge (\(|F|\geq 1\)). 
Two vertices \(H_1=(S_1,F_1)\) and \(H_2=(S_2,F_2)\) are adjacent in \(\Gamma\) if and only if \(S_1\cap S_2\neq\emptyset\). 
Thus an independent set in \(\Gamma\) corresponds to a collection of pairwise vertex-disjoint connected subgraphs of \(G\).

\paragraph{Weights.}
For a vertex \(H=(S,F)\) of \(\Gamma\), its weight \(w^{\sigma}(H)\) is defined as
\[
w^{\sigma}(H):= \frac{z^{|F|}\,Z_H^{\sigma}(A-J,\xi)}
{\Bigl(\prod_{v\in S\setminus\Lambda}\sum_{i=1}^q\xi_{v,i}\Bigr)\cdot\prod_{v\in\Lambda\cap S}\xi_{v,\sigma(v)}},
\]
where \(Z_H^{\sigma}(A-J,\xi)\) is the partition function on \(H\) with edge matrix \(A-J\) and external field \(\xi\), under the boundary condition \(\sigma\) restricted to \(\Lambda\cap S\). 
Define also the prefactor
\[
p^{\sigma}(\xi):= \Bigl(\prod_{v\in V\setminus\Lambda}\sum_{i=1}^q\xi_{v,i}\Bigr)\cdot\prod_{v\in\Lambda}\xi_{v,\sigma(v)}.
\]

\paragraph{Polymer representation.}
Lemma~8 of~\cite{Guus2021zerofreetossm} states that
\[
p^{\sigma}(\xi)\,Z_{\Gamma}(w^{\sigma}) = Z_G^{\sigma}(J+z(A-J),\xi),
\]
where \(Z_{\Gamma}(w^{\sigma})\) is the multivariate independence polynomial of \(\Gamma\) evaluated at the weights \(w^{\sigma}(H)\):
\[
Z_{\Gamma}(w^{\sigma}) = \sum_{\substack{\mathcal{I}\subseteq V(\Gamma)\\\mathcal{I}\text{ independent}}}
\;\prod_{H\in\mathcal{I}} w^{\sigma}(H).
\]

\paragraph{Cluster expansion of the ratio.}
Applying the formal cluster expansion (7) to \(\log Z_{\Gamma}(w^{\sigma})\) and differentiating with respect to \(\xi_{v,i}\), one obtains the following series expansion for the ratio (see~\cite[Lemma~9]{Guus2021zerofreetossm}):
\[
P_{G,v,i;A}^{\sigma}(z) = \frac{1}{q} + \sum_{\ell\geq 1} z^{\ell} \sum_{k\geq 1} \frac{1}{k!} 
\sum_{(H_1,\dots,H_k)\in \mathcal{C}_{v,\ell,k}(G)}
\phi\bigl(\Gamma(H_1,\dots,H_k)\bigr)\,
\frac{\partial}{\partial\xi_{v,i}} \prod_{j=1}^{k} \widehat{w}^{\sigma}(H_j)\Big|_{\xi=1},
\]
where
\begin{itemize}
    \item \(\mathcal{C}_{v,\ell,k}(G)\) is the collection of sequences \((H_1,\dots,H_k)\) of connected subgraphs of \(G\) with at least two vertices satisfying
    \(\sum_{j=1}^{k}|E(H_j)|=\ell\), \(v\in\bigcup_{j=1}^{k}V(H_j)\), and the graph \(\Gamma(H_1,\dots,H_k)\) (the subgraph of \(\Gamma\) induced by \(\{H_1,\dots,H_k\}\)) is connected;
    \item \(\phi\) is the Ursell function; 
    \item \(\widehat{w}^{\sigma}(H)=w^{\sigma}(H)\,z^{-|E(H)|}\) are the scaled weights.
\end{itemize}
The crucial observation is that the \(\ell\)-th coefficient of this series depends only on the subgraphs \(H_j\) whose vertex sets lie within distance \(\ell\) of \(v\) in \(G\). 
Consequently, the coefficients of \(z^{0},z^{1},\dots,z^{d-1}\) are identical for two boundary conditions that agree on the \(d\)-neighborhood of \(v\), yielding the LDC property.

\subsection{Limitations of the cluster expansion approach for general 2‑spin systems}

Despite its effectiveness for the restricted setting considered in~\cite{Guus2021zerofreetossm}, the cluster expansion formulation for graph homomorphisms suffers from two fundamental limitations that prevent its application to general 2‑spin systems.

\begin{enumerate}
    \item \textbf{Restricted parameter path.}
    The expansion is derived for a partition function whose edge matrix varies along the specific one‑parameter trajectory \(A(z)=J+z(A-J)\). 
    In other words, the zero‑freeness assumption must hold for a whole family of matrices that interpolate linearly between the all‑ones matrix \(J\) and the target matrix \(A\). 
    For a general 2‑spin system with edge activities \(\beta\) and \(\gamma\), the analogous statement would require zero‑freeness for all parameters lying on a line segment connecting \((\beta,\gamma)=(1,1)\) to the desired values. 
    This is far more restrictive than the univariate zero‑free regions (e.g., in \(\lambda\) alone or in \(\beta\) alone).
    
    \item \textbf{Absence of expansion in the external field variables.}
    The cluster expansion used above treats the auxiliary vertex weights \(w^{\sigma}(H)\) as power series in the single variable \(z\), while the external field parameters \(\xi_{v,i}\) are merely spectators. 
    No expansion is performed in the \(\xi\) variables themselves. 
    Consequently, the method does \emph{not} yield a multivariate series expansion of the ratios as functions of several independent complex parameters. 
    In a general 2‑spin system, one may need to expand simultaneously in \(\lambda\), \(\beta\), and \(\gamma\) to obtain LDC around different points (e.g., \(\lambda=0\) for arbitrary \(\beta,\gamma\), or \(\beta\gamma=1\) for arbitrary \(\lambda\)). 
    The cluster expansion approach, being inherently univariate in its current form, cannot yield such a unified multi‑parameter LDC statement. 
\end{enumerate}

In contrast, the Christoffel--Darboux type identity established in this paper makes no reference to any specific parameter trajectory and does not require a convergent series expansion in any of the variables. 
It directly factors the algebraic difference of partition functions with different pinned vertices, revealing powers of \(\lambda\) and \(\beta\gamma-1\) that depend only on the graph distance. 
This purely combinatorial identity therefore yields LDC uniformly for all three parameters and applies to zero‑free regions of arbitrary shape, overcoming the restrictions of the cluster expansion framework. 

\bibliographystyle{alpha}
\bibliography{new}
\end{document}